\documentclass[conference]{IEEEtran}

\usepackage{graphicx}
\usepackage{lipsum}
\usepackage{hyperref}
\usepackage{cite}
\usepackage{pifont}
\usepackage{enumitem}
\usepackage[font=small,skip=0pt]{caption}
\usepackage{subcaption}

\newtheorem{prop}{Proposition} 

\usepackage{adjustbox}
\usepackage{tabularx}
\IEEEoverridecommandlockouts
\usepackage{diagbox}
\usepackage{atbegshi}
\AtBeginDocument{\AtBeginShipoutNext{\AtBeginShipoutDiscard}}

\begin{document}

\title{\LARGE Pareto Optimal Hybrid Beamforming for Short-Packet Millimeter-Wave Integrated Sensing and Communication \vspace{-5mm}}

\author{\small
Jitendra~Singh,
Banda Naveen, \textit{Member,~IEEE,}
Suraj~Srivastava, 
Aditya~K.~Jagannatham, \textit{Senior Member,~IEEE}
\\and Lajos~Hanzo, \textit{Life Fellow,~IEEE} \vspace{-5mm}
}
\thanks{J. Singh B. Naveen, and A. K. Jagannatham are with the Department of Electrical Engineering, Indian Institute of Technology Kanpur, Kanpur, UP 208016, India (e-mail: jitend@iitk.ac.in; naveenb22@iitk.ac.in; adityaj@iitk.ac.in).
}
\thanks{S. Srivastava is with the Department of Electrical Engineering, Indian Institute of Technology Jodhpur, Jodhpur, Rajasthan 342030, India (email: surajsri@iitj.ac.in).}
\thanks{L. Hanzo is with the School of Electronics and Computer Science, University of Southampton, Southampton SO17 1BJ, U.K. (e-mail: lh@ecs.soton.ac.uk).
}

\maketitle
\begin{abstract}
Pareto optimal solutions are conceived for radar beamforming error (RBE) and sum rate maximization in short-packet (SP) millimeter-wave (mmWave) integrated sensing and communication (ISAC). Our ultimate goal is to realize ultra-reliable low-latency communication (uRLLC) and real-time sensing capabilities for 6G applications. The ISAC base station (BS) transmits short packets in the downlink (DL) to serve multiple communication users (CUs) and detect multiple radar targets (RTs). We investigate the performance trade-off between the sensing and communication capabilities by optimizing both the radio frequency (RF) and the baseband (BB) transmit precoder (TPC), together with the block lengths. The optimization problem considers the minimum rate requirements of the CUs, the maximum tolerable radar beamforming error (RBE) for the RTs, the unit modulus (UM) elements of the RF TPC, and the finite transmit power as the constraints for SP transmission. The resultant problem is highly non-convex due to the intractable rate expression of the SP regime coupled with the non-convex rate and UM constraints. To solve this problem, we propose an innovative two-layer bisection search (TLBS) algorithm, wherein the RF and BB TPCs are optimized in the inner layer, followed by the block length in the outer layer. Furthermore, a pair of novel methods, namely a bisection search-based majorizer and minimizer (BMM) as well as exact penalty-based manifold optimization (EPMO) are harnessed for optimizing the RF TPC in the inner layer. Subsequently, the BB TPC and the block length are derived via second-order cone programming (SOCP) and mixed integer programming methods, respectively. Finally, our exhaustive simulation results reveal the effect of system parameters for various settings on the RBE-rate region of the SP mmWave ISAC system and demonstrate a significantly enhanced performance compared to the benchmarks.

\end{abstract}

\maketitle
\section{\uppercase{INTRODUCTION}}
\IEEEPARstart{N}{e}xt-generation (NG) wireless networks aim for providing ultra-reliable low-latency connectivity (uRLLC), which supports challenging applications such as smart grids, industrial automation, autonomous vehicles, and mission-critical communication \cite{SPC_new_1,sutton2019enabling}. Short packet communication (SPC) is a key enabler for realizing uRLLC. However, the intractable expression of the achievable rate corresponding to the finite block length and the decoding error probability requirements in the SPC regime renders the beamforming optimization problems (OPs) intractable in SPC-aided wireless systems \cite{SPC_1,SPC_6,SPC_5}.

Recently, integrated sensing and communication (ISAC) in conjunction with millimeter wave (mmWave) technology has gained significant attention due to its excellent ability to provide sensing and communication (SC) capabilities in NG networks \cite{liu2020joint,10012421,9829746}. The similarity between the channel characteristics and signal processing tasks encountered both in sensing and communication pave the way for their integration in the existing cellular infrastructure while necessitating only moderate hardware changes \cite{mm_ISAC_1,mm_ISAC_2,mm_ISAC_3}. 
Moreover, to overcome the prohibitive requirement of a dedicated radio frequency (RF) chain for each antenna element mandated by the conventional architecture, the hybrid beamforming (HBF) approach that requires significantly fewer RF chains (RFCs) offers a viable alternative for practical realization of mmWave multiple-input and multiple-output (MIMO) systems \cite{mm_ISAC_4,mm_1}. More specifically, in the HBF scheme, the transmit precoder (TPC) is divided into the baseband (BB) and RF TPC, where the RF TPC is implemented by digitally controllable phase shifters.
Moreover, to study the SC trade-off in ISAC mmWave systems, Pareto optimization-based beamformer design is the ideal method of analyzing the optimal boundary of the SC performance \cite{ISAC_3,ISAC_4,P_3}. 

However, SPC transmission must be harnessed for supporting uRLLC services in industrial automation, autonomous vehicles, and mission-critical communication, where real-time communication and sensing play a vital role. In a similar fashion, smart cities can leverage uRLLC and ISAC to optimize traffic management systems, ensuring prompt responses to fluctuating traffic conditions and thus enhancing overall urban mobility.
Therefore, ISAC requires an additional layer of intelligence for combining the sensing capabilities with uRLLC services via mmWave communication. Thus, the large number of compelling applications gives rise to an increased number of systems requiring uRLLC communications combined with accurate and robust sensing capabilities \cite{saad2019vision},\cite{letaief2019roadmap}.
Inspired by these trends, we investigate SPC-enabled mmWave ISAC, which has the potential of significantly improving the overall performance of wireless networks. The optimization of the hybrid beamformer, along with the block length, plays a crucial role in supporting uRLLC services for communication users (CUs) and sensing for the radar targets (RTs) in SPC-enabled mmWave ISAC systems. 
Specifically, we characterize the trade-off between the sensing and communication tasks via the Pareto optimization of hybrid beamformers and the block length of an SPC-aided mmWave system. 
To the best of our knowledge, this is the first paper exploring the paradigm of SPC in an mmWave MIMO system, which optimizes the HBF and block length to meet the uRLLC requirements of multiple CUs, while also reliably sensing multiple RTs.
The next subsection presents a comprehensive literature survey in the area of SPC-aided mmWave systems.



\subsection{Literature review} \label{literature review}
The authors of the seminal papers \cite{ISAC_4,ISAC_3,ISAC_1,ISAC_2,ISAC_5,ISAC_6,ISAC_7,mm_ISAC_6} investigated the trade-off between the performance of SC by exploring the Pareto region of latency-agnostic ISAC-aided systems. Specifically, the authors of \cite{ISAC_1,ISAC_2} conceived cutting-edge techniques for the optimization of the transmit waveform to characterize the trade-off between the SC performance in an ISAC-enabled MIMO system, while considering both shared and separated antenna scenarios. 
Cao, in the treatise \cite{ISAC_5}, proposed a beamforming strategy for determining the Pareto boundary of ISAC-aided MIMO systems, while considering the sensing- and communication-SINRs as the metrics for the SC trade-off.
As a further advance, the authors of \cite{ISAC_4} consider the Cramér-Rao bound (CRB)-rate region of the SC trade-off in ISAC systems. 
The authors of \cite{ISAC_6} developed a revolutionary framework for Pareto optimal beamforming optimization of MIMO ISAC systems with the aim of analyzing the trade-off between the sensing and communication rates. 
As a further advance, Zou \textit{et al.} \cite{ISAC_7} explored the Pareto boundary for maximization of the energy efficiency (EE) in ISAC systems. The framework of \cite{ISAC_7} presented a novel constrained Pareto optimization problem (OP) for the maximization of the EE of the CUs, while constraining the sensing-centric EE. To solve this non-convex problem, an iterative successive convex approximation (SCA)-based algorithm is proposed to obtain the approximate Pareto boundary by evaluating a sequence of constrained problems subject to sensing-centric EE thresholds. Moreover, Yu \textit{et al.,} \cite{mm_ISAC_6} proposed a majorization and minimization (MM)-based algorithm for optimizing the beamforming in an ISAC-aided wireless system.

It is important to note that the fully-digital beamforming schemes discussed in \cite{ISAC_3,ISAC_4,ISAC_1,ISAC_2,ISAC_5,ISAC_6,ISAC_7} are inefficient for ISAC-aided mmWave MIMO systems, since they require an excessive number of RFCs. To this end, the related expositions \cite{mm_1,MIP_1,mm_ISAC_5,mm_ISAC_1,mm_ISAC_2,he2022qcqp,mm_ISAC_3,mm_ISAC_4,HBF_8} conceived HBF designs for ISAC-aided mmWave MIMO systems, which significantly reduces the number of RFCs and yet performs close to the fully-digital schemes. Specifically, the authors of \cite{mm_1} proposed an innovative orthogonal matching pursuit (OMP)-based algorithm for optimizing the BB and RF TPCs.
As a further advance, the authors of \cite{mm_ISAC_5} present innovative techniques for minimizing the weighted radar beamforming error (RBE) and communication beamforming error to optimize both the RF and BB TPCs of the HBF scheme. Furthermore, the authors of \cite{MIP_1} proposed a mixed integer programming-based HBF for mmWave MIMO systems. 
Along similar lines, the authors of \cite{mm_ISAC_2} proposed an HBF scheme for the ISAC-aided mmWave MIMO systems based on the Riemannian conjugate gradient (RCG) method for handling the UM constraint of RF TPC, which aims to maximize the achievable data rate of the CUs, while achieving the accurate sensing of the RTs. As a further advance, the authors of \cite{mm_ISAC_3} consider a partially connected hybrid architecture for ISAC-aided mmWave MIMO systems, where they focused on the dual objectives of minimizing the Cramér-Rao bound (CRB) for the estimation of the direction of arrival (DOA) and maximizing the signal-to-interference-plus-noise ratio (SINR) of the received radar echos.

The transmission models underlying a large fraction of the literature surveyed above, namely \cite{ISAC_1,ISAC_2,ISAC_5,ISAC_4,ISAC_3,mm_ISAC_2,mm_ISAC_3,mm_ISAC_4} are based on communication with infinite block lengths (IBLs), hence they rely on the conventional Shannon capacity formula. While their analyses are ingenious and immensely useful in their specific settings, these models are agnostic of the stringent reliability and latency requirements of uRLLC applications.
More specifically, the classical Shannon capacity formula, that considers IBL transmission, is inapplicable in the SPC regime due to the finite block length and non-zero code word error probability specifications.
Thus, following the innovative rate expression provided by \cite{SPC_1}, the authors of \cite{SPC_2,SPC_3,SPC_4,SPC_5,SPC_6} proposed inspiring beamforming designs for SPC-enabled systems. Specifically, He \textit{et al.} \cite{SPC_2} derived novel beamforming techniques for the SPC-enabled multi-CU (MU) multiple-input and single-output (MISO) downlink, where they addressed the optimization of multiple objectives including the weighted sum rate, EE, and CU fairness, while also considering the minimum CU-rate and transmission power constraints. Furthermore, the authors of \cite{SPC_3} presented state-of-the-art beamformer designs for an SPC-enabled MIMO system using alternating optimization and fractional programming. Their innovative beamforming techniques maximize the achievable data rate for a given transmit power budget.
On the other hand, the authors of \cite{SPC_4} designed efficient resource allocation and beamforming algorithms for SPC-enabled MU-MISO systems that achieve transmit power minimization, while constraining the decoding-error probability of SPC transmission. 
As a further advance, Huang \textit{et al.} \cite{SPC_6}investigated the rate region of the SPC-enabled MISO interference channel and optimized the beamformer weights and block length to maximize the sum rate of the system, considering the resource allocation and block length as constraints.

Following the above discussion of ISAC and SPC in separate contexts, we now move our focus to the literature that explored their integration. 
The authors of the path-breaking works \cite{SPC_ISAC_1,SPC_ISAC_2,SPC_ISAC_3} proposed a transmit beamforming paradigm for SPC-enabled ISAC systems, wherein the base station (BS) performs detection of an RT and provides uRLLC services for the CUs. 
The pioneering research in \cite{SPC_ISAC_1} presents a framework for transmit power minimization, while meeting the critical radar sensing and uRLLC latency requirements. The authors of \cite{SPC_ISAC_1} proposed a creative quadratic transform-based fractional programming approach in conjunction with an interior point method to solve the pertinent OP. 
To explore the trade-off between uRLLC data transmission and target localization in SPC-aided ISAC systems, Zabini \textit{et al.} derived a novel beamforming scheme for minimizing the CRB of RT localization under constraints on the block error probability for the communication CUs in \cite{SPC_ISAC_2}.
Along similar lines, in the avant-grade investigation of \cite{SPC_ISAC_3}, the researchers succeeded in developing a joint beamforming and scheduling scheme for an SPC-enabled ISAC system that meets the stringent uRLLC requirements of the CUs.  

Moreover, considering the strong possibility that next-generation networks may harness mmWave technology in view of the imminent spectrum crunch in the sub-6 GHz band, it is crucial to explore the SC trade-off in SPC-enabled mmWave ISAC systems. 
Exploring the Pareto boundary of such systems holds the key for SPC. Clearly, there is a paucity of
SPC-enabled mmWave ISAC system studies in the open literature. Explicitly, in the complex face of challenges, such as the complex rate expression of SPC transmission, integration of sensing and communication tasks, coupled with the hybrid design of the TPC, the associated beamforming optimization is not well documented. 
Inspired by this knowledge-gap, we conceive a novel HBF scheme to achieve Pareto optimal SC trade-offs in SPC-enabled mmWave ISAC systems. 
The main contributions of this paper are enumerated next.

\subsection{Contributions of this work}\label{contributions}
\begin{enumerate}
    \item To begin with, we consider an SPC-enabled mmWave ISAC system, where an ISAC BS transmits the SPC-encoded signal to serve multiple CUs and detect the multiple RTs present. To reveal the trade-off between SC, we determine the RBE-rate region of the system, where the RBE and rate serve as metrics for sensing and communication. 
    \item We formulate the OP for the Pareto boundary of the RBE-rate region, while considering the minimum rate requirement for the CUs and maximum tolerable RBE for the RTs as constraints. Additional constraints arise due to the hybrid MIMO architecture, limited transmit power, and finite block length arising due to the SPC regime. Naturally, the problem thus formulated is highly non-convex due to the intractable rate expression of the SPC regime coupled with the non-convex constraints. 
     \item To solve the above problem, we propose an iterative two-layer bisection search (TLBS) algorithm, where the inner layer minimizes the RBE of the system by optimizing the BB and RF TPCs for a fixed sum rate, and subsequently, the outer layer optimizes the block length and updates the achievable rate via the bisection search method.
    \item More specifically, in the inner layer, we reformulate the rate expression to transform the intractable QoS constraint into a tractable SINR constraint and, subsequently, adopt the BCD principle for iteratively optimizing the BB and RF TPCs. To optimize the RF TPC, we propose two novel methods: bisection-based majorization and minimization (BMM) as well as exact penalty-based manifold optimization (EPMO). Following this, the SOCP method is employed to optimize the BB TPC.    
  Furthermore, we update the block length and sum rate via mixed integer programming and a bisection search in the outer layer.
  \item Finally, we evaluate the performance of the proposed scheme by evaluating the Pareto boundaries, achievable sum rates, and beam patterns for a variety of settings and compare them to the benchmarks for verifying the effectiveness of our proposed algorithm.

\end{enumerate}
\subsection{Notation}\label{notation}
We use the following notations throughout the
paper: $\mathbf{A}$,  $\mathbf{a}$, and $a$ represent a matrix, a vector, and a scalar quantities, respectively.
The $(i,j)$th element, and Hermitian of matrix $\mathbf{A}$ are denoted by  $\mathbf{A}{(i,j)}$, and $\mathbf{A}^H$, respectively. The trace, Frobenius
norm and vectorization of a matrix $\mathbf{A}$ are represented as $\mathrm{tr}\left(\cdot\right)$, $\left\vert\left\vert \mathbf{A} \right\vert\right\vert_F$ and $\mathrm {vec}\left (\cdot\right)$. The expectation operator is represented as $\mathbf{E}\{\cdot\}$; the real part of a quantity is denoted by ${\mathrm{Re}\left (\cdot\right)}$. ${\mathbf I}_M$ denotes an $M \times M$ identity matrix; the symmetric complex Gaussian distribution of mean $\mathbf{\mu}$ and covariance matrix $\mathbf{\sigma^2}$  is represented as ${\cal CN}(\mathbf{\mu}, \mathbf{\sigma^2})$. The operators $\odot$ and $\otimes$ denote the Hadamard product and  Kronecker product respectively, while ${\mathcal W}\left(\cdot,\cdot ; \cdot \right)$ represents the generalized Lambert function. The elements in $\arg[\mathbf{z}]$ are the phases of the input complex vectors.
\section{System Model and Problem Formulation of SPC-aided mmWave ISAC Systems}\label{System Model}

        
As shown in Fig. \ref{fig:sys_1}, we consider an SPC-enabled ISAC downlink operating in the mmWave band, where an ISAC BS communicates with $M$ multi-uRLLC CUs and detects $N_\mathrm{tar}$ RTs, simultaneously. 
The ISAC BS relies on a fully-connected hybrid MIMO architecture having $N_\mathrm{t}$ transmit antennas and $N_\mathrm{RF}\leq N_\mathrm{t}$ RF chains, as shown in Fig. \ref{fig:sys_2}. Moreover, each CU is equipped with a single antenna. Therefore, to support simultaneous service for $M$ single antenna CUs while detecting $N_\mathrm{tar}$ RTs through SPC transmission at the angles of interest, one has to follow the condition $M \leq N_\mathrm{RF}<< N_\mathrm{t}$.


\subsection{mmWave channel model}
This paper employs the popular Saleh-Valenzuela channel model \cite{mm_ISAC_1,mm_ISAC_2,mm_ISAC_3} to capture the geometric properties of the mmWave channel, which includes complex-valued path losses, angles-of-arrival (AoAs), and angles-of-departure (AoDs) arising due to the paths scattered by a finite number of dominant clusters. Mathematically, the frequency-flat mmWave MISO channel $\mathbf{h}^H_m \in \mathbb{C}^{N_\mathrm{t} \times 1}$ between the ISAC BS and the $m$th CU is expressed as
\begin{equation}\label{eqn:channel}
\mathbf{h}^H_m=\sqrt \frac {N_{t}}{N_\mathrm{clu}N_\mathrm{ray}} \sum_{i=1}^{N_\mathrm{clu}}\sum _{j=1}^{N_\mathrm{ray}}\alpha^m_{i,j}\mathbf{a}^H_\mathrm{BS}(\phi^m_{i,j}),
\end{equation}
where $N_\mathrm{clu}$ and $N_\mathrm{ray}$ are the number of scattering clusters and scattered rays per cluster, respectively. The quantity $\alpha^m_{i,j}$ in (\ref{eqn:channel}) is the multipath channel gain that is distributed as $\mathcal{CN}(0, 10^{-0.1PL(d_m)}), \forall i =\{1,\hdots, N_\mathrm{clu}\}$, and $j = \{1,\hdots, N_\mathrm{ray}\}$ where $PL(d_m)$ is the path loss in dB that depends on the distance $d_m$ associated with the corresponding link.
Moreover, we consider that the ISAC BS is employed with a uniform linear array (ULA), owing to the array response vector $\mathbf{a}_\mathrm{BS}(\phi^m_{i,j}) \in \mathbb{C}^{N_\mathrm{t}\times 1}$ as
\begin{equation} 
\mathbf{a}_\mathrm{BS}(\phi^m_{i,j})=\frac {1}{\sqrt{N_\mathrm{t}}}\left[{1,e^{j\frac {2\pi \bar {d}}{\lambda}\sin(\phi^m_{i,j})},\ldots, e^{j(N_\mathrm{t}-1)\frac {2\pi \bar {d}}{\lambda }\sin(\phi^m_{i,j})}}\right]^T,
\end{equation}
where $\phi^m_{i,j}$ denotes the AoD, $\lambda$ is the carrier wavelength, and $d$ is the spacing between adjacent antennas, which is set as $d = \lambda/2$.

\begin{figure}[t]
\setkeys{Gin}{width=\linewidth}
   \hspace{-8mm}
    \begin{subfigure}[t]{0.25\textwidth}
    \vspace{-1cm}
    \includegraphics[width=1.1\textwidth]{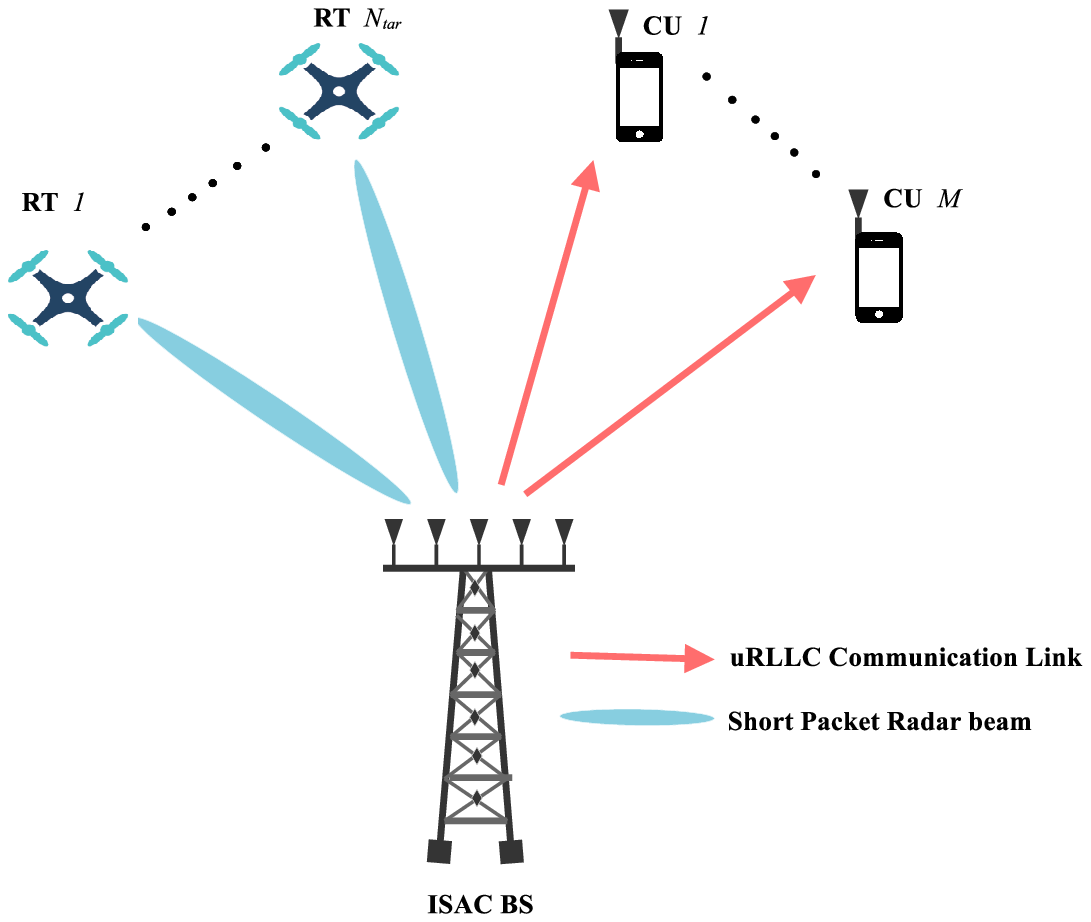}
    \vspace{-2cm}
    \caption{} \label{fig:sys_1}
\end{subfigure}
\begin{subfigure}[t]{0.25\textwidth}
\vspace{-1.2cm}
    \includegraphics[width=1.2\textwidth]{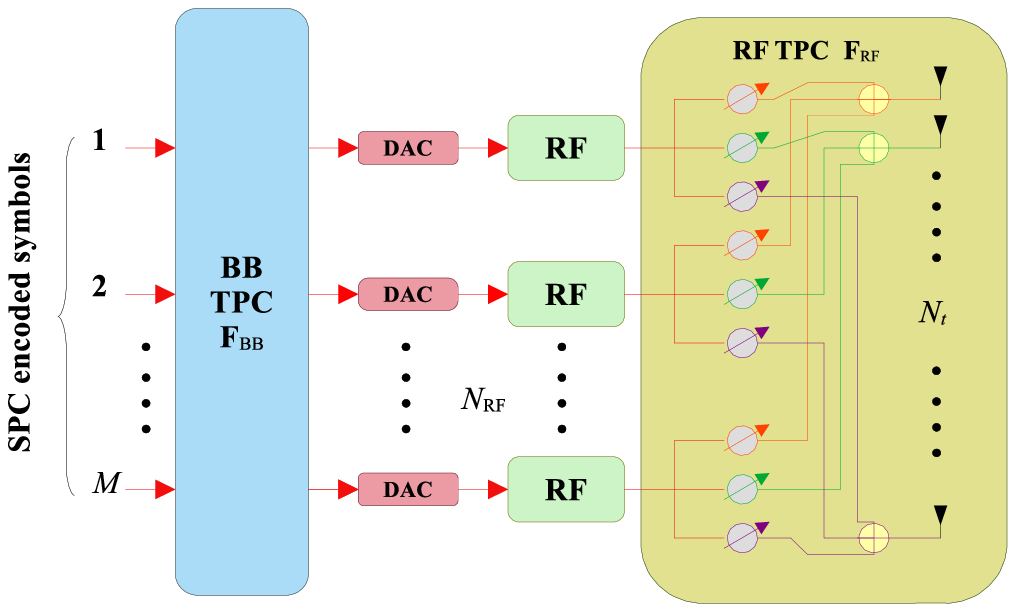}
    \vspace{-2.4cm}
    \caption{} \label{fig:sys_2}
\end{subfigure}
\caption{(a) Illustration of an uRLLC mmWave ISAC system. (b) Block diagram of hybrid beamforming architecture at the ISAC BS.}
\vspace{-0.7cm}
\end{figure}

\subsection{Signal model}
Let us consider that the data of each CU is encoded with an individual encoder of finite block length $\beta_m, m = 1, \hdots, M$ at the ISAC BS. Thus, the information bits of the uRLLC CUs are transmitted via packets. In a similar fashion, each CU decodes its data independently, considering a non-zero decoding error probability of $\epsilon_m, m = 1, \hdots, M$. 
Furthermore, the encoded symbols of all the CUs at the ISAC BS are initially processed by the BB TPC $\mathbf{F}_\mathrm{BB}=[\mathbf{f}_{{\rm BB},1}, \hdots,\mathbf{f}_{{\rm BB}, M}]\in {\mathbb C}^{{N_\mathrm{RF}} \times M}$, followed by the RF TPC $\mathbf{F}_\mathrm{RF}\in {\mathbb C}^{{N_\mathrm{t}} \times {N_\mathrm{RF}}}$.
Thus, the downlink transmitted signal $\mathbf {x}\in \mathbb{C}^{N_\mathrm{t}\times 1}$ from the ISAC BS is given by
\begin{equation} \label{eqn:tx_signal}
\mathbf {x}=\mathbf{F}_{\rm RF}\mathbf{F}_{\rm BB}\mathbf{s} = \mathbf{F}_{\rm RF}\sum_{m=1}^M\mathbf{f}_{{\rm BB},m}s_m,
\end{equation}
where $\mathbf{s}=[s_1,s_2,\hdots,s_M]^T\in {\mathbb C}^{{M} \times 1}$ is the encoded signal, which serves as a combined signal for the radar detection as well as downlink communication transmission \cite{mm_ISAC_2}. 
Moreover, the encoded symbols are assumed to be statistically independent and identically distributed (i.i.d), which satisfy $\mathbb{E}\{\mathbf{s}\} = \mathbf{0}$ and $\mathbb{E}\{\mathbf{s}\mathbf{s}^H\}=\mathbf{I}_M$. 
Consequently, the covariance matrix $\mathbf{C}_\mathrm{x} \in \mathbb{C}^{N_\mathrm{t}\times N_\mathrm{t}}$ of the transmitted signal $\mathbf{x}$ is given by
\begin{equation} 
\mathbf{C}_\mathrm{x}=\mathbf{E}\{\mathbf{x}\mathbf{x}^H\} = \mathbf{F}_{\rm RF}\mathbf{F}_{\rm BB}\mathbf{F}_{\rm BB}^H\mathbf{F}_{\rm RF}^H.
\end{equation}
\subsection{Radar model}
This paper considers the mono-static MIMO radar used at the ISAC BS for target detection, where the same antenna arrays are used for transmitting and receiving radar signals. Thus, the echo signal $\mathbf{y}_\mathrm{rad} \in \mathbb{C}^{N_\mathrm{t} \times 1}$ received at the ISAC BS can be written as \cite{mm_ISAC_2}
\begin{equation} \label{eqn:rad_rec}
\begin{aligned}
\mathbf{y}_\mathrm{rad}=
&\sum_{t=1}^{N_\mathrm{tar}}{\zeta}_t^\mathrm{tar}\mathbf{a}_\mathrm{BS}({\theta}_t^\mathrm{tar})\mathbf {a}_\mathrm{BS}^{H}({\theta}_t^\mathrm{tar})\mathbf{x}\\
+&\sum_{c=1}^{N_\mathrm{ct}}{\zeta}_c^{\mathrm{ct}}\mathbf{a}_\mathrm{BS}({\theta}_c^\mathrm{ct})\mathbf {a}_\mathrm{BS}^{H}({\theta}_c^\mathrm{ct})\mathbf{x}+ \mathbf {n}_\mathrm{rad},
\end{aligned}
\end{equation}
where the first and second terms in (\ref{eqn:rad_rec}) are the desired target signal and the echo signal due to clutter, respectively, and $\mathbf{n}_\mathrm{rad} \in \mathbb{C}^{N_\mathrm{t} \times 1}$ is the noise encountered in the radar sensing environment. The quantities ${\zeta}_t^\mathrm{tar}$ and ${\zeta}_c^{\mathrm{ct}}$ are the complex-valued path loss reflection coefficients of the RTs and clutters located at angles ${\theta}_t^\mathrm{tar}$ and at ${\theta}_c^\mathrm{ct}$, respectively. Based on the received signal (\ref{eqn:rad_rec}), one can estimate the angles of the RTs by employing the well-known MUSIC algorithm \cite{liu2020joint}. Assuming perfect estimation of the target angle $\theta$, the transmit beampattern gains $G(\theta)$ can be expressed as
\begin{equation}
G(\theta) = \mathbf {a}_\mathrm{BS}^{H}(\theta )\mathbf{C}_\mathrm{x}\mathbf {a}_\mathrm{BS}(\theta).
\end{equation}
Note that $G(\theta)$ is the spatial beam pattern, which has to be synthesized for the target sensing environment. Observe that designing $G(\theta)$ is equivalent to designing the covariance matrix $\mathbf{C}_\mathrm{x}$. Hence, one has to design $\mathbf{F}_\mathrm{RF}$ and $\mathbf{F}_\mathrm{BB}$, that meet the sensing requirements of the RTs. Therefore, to optimize the radar sensing performance, we design $\mathbf{C}_\mathrm{x}$ to approach the ideal desired radar covariance matrix $\mathbf{C}_d = \mathbf{F}_\mathrm{r}\mathbf{F}_\mathrm{r}^H$, where $\mathbf{F}_\mathrm{r}\in {\mathbb C}^{{N_\mathrm{t}} \times {N_\mathrm{tar}}}$ is the ideal radar beamformer used for the RTs, which is given by
\begin{equation}\label{eqn:des_rad}
\mathbf{F}_\mathrm{r}= \left[\mathbf {a}_\mathrm{BS}(\theta^\mathrm{tar}_1), \mathbf {a}_\mathrm{BS}(\theta^\mathrm{tar}_2),\hdots,\mathbf {a}_\mathrm{BS}(\theta^\mathrm{tar}_{N_{tar}})\right].
\end{equation}
To evaluate the performance of the sensing, we consider the radar beamforming error (RBE) of the RTs denoted by $\mathcal{E}$ \cite{mm_ISAC_5}, which is given by
\begin{equation}
\mathcal{E}\left(\mathbf{F}_\mathrm{RF}, \mathbf{F}_\mathrm{BB}, \mathbf{U}\right) = \|\mathbf{F}_\mathrm{RF}\mathbf{F}_\mathrm{BB}-\mathbf{F}_\mathrm{r}\mathbf{U}\|_F^2,
\end{equation}
where $\mathbf{U} \in \mathbb{C}^{N_\mathrm{tr} \times M}$ is an auxiliary unitary matrix obeying $\mathbf{U}\mathbf{U}^H = \mathbf{I}_{N_\mathrm{tar}}$.
\subsection{Communication model}
Based on the mmWave MISO channel (\ref{eqn:channel}) and the transmission signal model (\ref{eqn:tx_signal}), the received signal $y_m$ at the $m$th CU can be written as
\begin{subequations}
\begin{align}
{y}_{m}=&\mathbf{h}_m^H\mathbf{F}_{\rm RF}\mathbf{F}_{\rm BB}\mathbf{s} + n_m \\
=&{\bf h}_{m}^{\rm H}\mathbf{F}_{\rm RF}\mathbf{f}_{{\rm BB},m}s_m+\sum_{n=1, n \neq m}^M\hspace{-0.3cm}\mathbf {h}_m^H \mathbf{F}_{\rm RF}\mathbf{f}_{{\rm BB},n} s_n+ n_m,
\end{align}
\end{subequations}
where $n_m$ is the i.i.d. complex additive white Gaussian noise (AWGN) having the distribution $n_m\sim\mathcal{CN}({0}, N_{o})$. Consequently, the corresponding signal-to-interference-plus-noise ratio (SINR) $\gamma_m$ of the $m$th CU is evaluated as
\begin{equation}
\begin{aligned}
\mathbf{\gamma}_m =& \frac{\left\vert\mathbf{h}_m^H \mathbf{F}_{\rm RF}\mathbf{f}_{{\rm BB},m}\right\vert^2}{\sum_{n=1, n \neq m}^{M}{\left\vert\mathbf{h}_m^H \mathbf{F}_{\rm RF}\mathbf{f}_{{\rm BB},n}\right\vert^2} +N_{o}}.
\end{aligned}
\end{equation}
Following Shannon's capacity formula, the maximum achievable transmission rate $S_m$ of this CU, in nats/s/Hz/channel, is given by 
\begin{equation}\label{eqn:shanon}
S_m=\ln{\left(1+\gamma_m\right)},\forall m.
\end{equation}
However, the conventional Shannon capacity relation holds true only for IBL transmission, wherein the error probability tends to zero. Since this paper considers practical SPC, the achievable rate given by (\ref{eqn:shanon}) is not a realistic model. In this context, thanks to the results in \cite{SPC_1,SPC_2,SPC_4,SPC_5}, the achievable rate $R_m$ of the $m$th uRLLC CU owing to the SPC transmission with a finite block length of $\beta_m$ and transmission error probability $\epsilon_m$ is given by 
\begin{equation} \label{rate_expre}
R_m=\ln{\left(1+\gamma_m\right)}-\sqrt{\frac{V_m}{\beta_m}} {Q^{ - 1}}\left({{\epsilon_{m}}} \right),\forall m,
\end{equation}
where $Q^{-1}(.)$ is the inverse of the Gaussian Q-function and $V_m$\footnote{It gauges the variability of the channel relative to a deterministic bit pipe with the same capacity.} represents the channel dispersion of the uRLLC CU $m$, which is given by
\begin{equation}
V_m\left ({\gamma_m}\right )=1-\frac {1}{\left ({1+\gamma_m }\right )^{2}}.
\end{equation}
Consequently, the achievable sum rate of the CUs is given by
\begin{equation} \label{sum_rate}
\mathcal{R}=\sum_{m=1}^M R_m.
\end{equation}
\subsection{Problem Formulation}\label{problem formulation}
This paper aims for jointly optimizing the Pareto optimal RF TPC $\mathbf{F}_\mathrm{RF}$, BB TPC $\mathbf{F}_\mathrm{BB}$ and the block lengths $\{\beta_m\}_{m=1}^M$ to characterize the RBE-rate region of an SPC-enabled ISAC mmWave system. 
Specifically, the RBE-rate region of the system under consideration is defined as the collection of all the feasible twin tuples $(\mathcal{E}, \mathcal{R})$ that can be simultaneously achieved, where $\mathcal{E}$ and $\mathcal{R}$ are sensing and communication metrics, respectively. 
Therefore, we are interested in evaluating the Pareto front \cite{ISAC_3,ISAC_4} constituted by all optimal twin tuples $( \mathcal{E}, \mathcal{R})$ in the boundary of the RBE-rate region.
Typically, the Pareto front consists of $( \mathcal{E}, \mathcal{R})$ pairs at which it is impossible to simultaneously improve the communication and sensing performance without a compromise between them. More specifically, for a given SPC-aided mmWave ISAC system, any $( \mathcal{E}, \mathcal{R})$ twin tuple located on the Pareto boundary of the rate-RBE region is formulated as
\begin{subequations}\label{OP:1}
\begin{align} 
\mathcal{P}_{0}: \hspace{4mm}&\mathop {\max }\limits_{\mathbf{F}_\mathrm{RF}, \mathbf{F}_\mathrm{BB},\mathbf{U}, \{\beta_m\}_{m=1}^M}\quad \quad  \mathcal{R}, \label{OF:1}\\
&\mathrm{s.\hspace{1mm} t.} \quad{R_{m}} \geq {\eta _{m}} \mathcal{R},\forall m,\label{cons:rate} \\
&\qquad\mathcal{E} \leq \mathcal{E}_\mathrm{max},\label{cons:radar}\\
&\qquad \mathbf{U}\mathbf{U}^H = \mathbf{I}_{N_\mathrm{tar}},\label{cons:UT}\\
&\qquad \left\vert\mathbf{F}_\mathrm{RF}(i,j)\right\vert = 1, \forall i, j, \label{cons:RF}\\ &\qquad\|\mathbf{F}_\mathrm{RF}\mathbf{F}_\mathrm{BB}\|_F^2\leq P_\mathrm{max}\label{cons:TP}\\ 
&\qquad \sum \nolimits_{m=1}^{M} \beta_{m} \leq N, \label{cons:BL}\\ 
&\qquad \beta_m \in {{\mathbb{Z}}^ + },\forall m, \label{cons:IN}
\end{align}
\end{subequations}
where (\ref{cons:rate}) denotes the QoS constraint for the CUs with $\eta_m\in (0,1)$ denoting the ratio between the achievable rate of the $m$th CU and the sum rate $\mathcal{R}$, such that $\sum_{m = 1}^M \eta_m =1$.
Furthermore, the constraint (\ref{cons:radar}) is the RBE required for the RTs, while 
(\ref{cons:RF}) is the unit modulus (UM) constraint due to the phase shifters of the hybrid architecture, and (\ref{cons:TP}) is the maximum transmit power budget constraint. Moreover, the last two constraints (\ref{cons:BL}) and (\ref{cons:IN}) are due to the SPC regime, where all the uRLLC CU symbols are transmitted within a maximum of $N$ symbols and the block length $\beta_m$ must be a non-negative integer.

The complete Pareto boundaries of the achievable RBE-rate region of the SPC-enabled mmWave ISAC system above can be characterized by solving the OP $\mathcal{P}_0$, which is however, challenging due to the non-convex constraints (\ref{cons:rate}) and (\ref{cons:RF}) and tightly coupled variables in the objective function and the constraints. Moreover, the rate expression defined by (\ref{rate_expre}) is more decimate than the traditional Shannon formula, which exacerbates the challenge. Therefore, the next section proposes a novel TLBS method that overcomes this obstacle by intelligently exploiting a bisection search. 
\section{Two layer bisection search for joint optimization}
In the proposed TLBS method, the inner layer evaluates $\mathbf{F}_\mathrm{RF}, \mathbf{F}_\mathrm{BB}$ and $\mathbf{U}$ to minimize the RBE $\mathcal{E}=\|\mathbf{F}_\mathrm{RF}\mathbf{F}_\mathrm{BB} - \mathbf{F}_\mathrm{r}\mathbf{U}\|_F^2$ for the SPC parameters $\beta_m, \forall m$, whereas the outer layer updates the achievable rate $\mathcal{R}$ by employing the bisection search.
Specifically, in the inner layer, we obtain $\mathbf{F}_\mathrm{RF}, \mathbf{F}_\mathrm{BB}$ and $\mathbf{U}$ for feasible values of $\mathcal{R}$ and subsequently update the sum rate $\mathcal{R}$ and $\{\beta_m\}_{m=1}^M$ in the outer layer. For any given sum rate $\mathcal{R} \geq 0$ and $\beta_m$, the equivalent feasible problem of $\mathcal{P}_0$ in the inner layer is the minimization of the RBE, which is given by
\begin{equation}\label{eqn:HBF_ISAC_6}
\begin{aligned} 
\mathcal{P}_{1}:  &\min \limits _{\mathbf{F}_\mathrm{RF}, \mathbf{F}_\mathrm{BB},  \mathbf{U}} ~  \mathcal{E}\left(\mathbf{F}_\mathrm{RF}, \mathbf{F}_\mathrm{BB}, \mathbf{U}\right)= ~\|\mathbf{F}_\mathrm{RF}\mathbf{F}_\mathrm{BB} - \mathbf{F}_\mathrm{r}\mathbf{U}\|_F^2 \\ 
&\text {s.t.}\quad \text{(\ref{cons:rate}), (\ref{cons:UT}), (\ref{cons:RF}), and (\ref{cons:TP})}.
\end{aligned}
\end{equation}
To obtain the optimal solution $\{{ \mathbf{F}^*_\mathrm{RF}, \mathbf{F}^*_\mathrm{BB},\mathbf{U}^*}\}$ in the inner layer, we solve the above OP (\ref{eqn:HBF_ISAC_6}).
Subsequently, if the OP (\ref{eqn:HBF_ISAC_6}) is feasible for the given $\{ \mathbf{F}_\mathrm{RF}, \mathbf{F}_\mathrm{BB},\mathbf{U}\}$, we perform a bisection search by solving a sequence of feasibility problems corresponding to the problem $\mathcal{P}_0$ to update the optimal value of the rate $\mathcal{R}^*$ and the block length $\{\beta^*_m\}_{m=1}^M$ in the outer layer.
Moreover, upon denoting the optimal solutions of $\mathcal{P}_1$ as $\mathbf{F}^*_\mathrm{RF}, \mathbf{F}^*_\mathrm{BB}$ and $\mathbf{U}^*$, it is evident that the OP $\mathcal{P}_1$ is feasible if $||\mathbf{F}^*_\mathrm{RF}\mathbf{F}^*_\mathrm{BB} - \mathbf{F}_\mathrm{r}\mathbf{U}^{*}\|_F^2 \leq \mathcal{E}_\mathrm{max}$, and $\|\mathbf{F}^*_\mathrm{RF}\mathbf{F}^*_\mathrm{BB}\|^2_F\leq P_\mathrm{max}$, otherwise, it is considered infeasible. 
\subsection{Inner layer: BCD algorithm for solving $\mathcal{P}_{1}$}
Because of the finite block length $\beta_m$ and the transmission error probability $\epsilon_m$ in the rate expression, the constraint (\ref{cons:rate}) is highly non-convex.
Therefore, to solve $\mathcal{P}_1$, we first transform the non-convex constraint (\ref{cons:rate}) into a tractable form by using the following proposition:
\begin{prop}
Given $\beta_{m}\geq0$ and $\epsilon_m\geq0, \forall m$, the constraint ${R_{m}} \geq {\eta _{m}}\mathcal{R}$ is equivalent to ${\gamma_{m}} \geq {\Gamma_{m}}, \forall m,$ where ${\Gamma_{m}}$ is given by 
\begin{equation} \label{gamma}
{\Gamma_{m}} = e^{\eta_{m}\mathcal{R} + \frac{\kappa_{m}}{2}} - 1, 
\end{equation}
and $\kappa_{m}$ is the generalized Lambert ${\mathcal W}$ function, which is given by \cite{mezHo2017generalization}
\begin{equation}
{\kappa_{m}} = {\mathcal W}\left({^{\frac{2{Q^{ - 1}\left(\epsilon_{m} \right)}}{{\sqrt{\beta}_{m} }},\frac{{ - 2{Q^{ - 1}}\left(\epsilon_{m} \right)}}{\sqrt{\beta}_m}}; - 4{\delta_m^{2}}{\left({\frac{{Q^{ - 1}}\left(\epsilon_m \right)}{{\sqrt{\beta_{m}} }}} \right)}^{2}} \right),  
\end{equation}
where $\delta_m =  e^{-\eta_{m}\mathcal{R}}$.
\end{prop}
\begin{proof}
Please refer to Appendix (\ref{Appendix:A}) for the detailed proof.
\end{proof}
Upon employing the above proposition, the OP $\mathcal{P}_1$ can be recast as follows 
\begin{subequations}\label{eqn:HBF_ISAC_8}
\begin{align} 
&\min \limits _{\mathbf{F}_\mathrm{RF}, \mathbf{F}_\mathrm{BB},  \mathbf{U}}~\|\mathbf{F}_\mathrm{RF}\mathbf{F}_\mathrm{BB} - \mathbf{F}_\mathrm{r}\mathbf{U}\|_F^2 \\ 
&\text {s. t.}\quad \gamma_m \geq \Gamma_m, \forall m, \label{cons:SINR_1}\\
&\quad\quad \text{(\ref{cons:UT}),(\ref{cons:RF}), and (\ref{cons:TP}) }.
\end{align}
\end{subequations}
Since the optimization variables $\mathbf{F}_\mathrm{RF}, \mathbf{F}_\mathrm{BB}$ and $\mathbf{U}$ are coupled in both the objective function and the constraints of (\ref{eqn:HBF_ISAC_8}), we adopt the BCD method to decouple the optimization variables $\mathbf{F}_\mathrm{RF}$, $\mathbf{F}_\mathrm{BB}$, and $\mathbf{U}$ in (\ref{eqn:HBF_ISAC_8}), which renders it easier to break down the intricate problem into distinct sub-problems, each of which is focused on maximizing a particular block of variables.
\subsubsection{Sub-problem for $\mathbf{F}_\mathrm{RF}$}
For fixed $\mathbf{F}_\mathrm{BB}$ and $\mathbf{U}$, the equivalent OP for the design of $\mathbf{F}_\mathrm{RF}$ is given by
\begin{subequations}\label{RF_1}
\begin{align} 
& \min \limits _{\mathbf{F}_\mathrm{RF}}~\|\mathbf{F}_\mathrm{RF}\mathbf{F}_\mathrm{BB}-\mathbf{F}_\mathrm{r}\mathbf{U}\|_F^2 ,\label{eqn:obj}\\ 
&\text {s. t.}\quad \gamma_m \geq \Gamma_m, \forall m, \qquad \text{(\ref{cons:RF}), and (\ref{cons:TP}) } .  
\end{align}
\end{subequations}
The OP (\ref{RF_1}) above is still non-convex due to the non-convex SINR constraints (\ref{cons:SINR_1}) and UM constraint (\ref{cons:RF}). To handle this issue, let us 
rewrite the SINR $\gamma_m$ of the $m$th CU as follows
\begin{equation}\label{cons:SINR_7}
\begin{aligned}
\mathbf{\gamma}_m =&\frac{{\mathrm{tr}\left({\mathbf{F}_{\rm RF}^H\mathbf{H}_{m}\mathbf{F}_{\rm RF}\mathbf{B}_{m}}\right)}}{{\sum \limits_{\begin{subarray}{l} n=1, n \neq m \end{subarray}}^{M} \mathrm{tr}\left({\mathbf{F}_{\rm RF}^H\mathbf{H}_{m}\mathbf{F}_{\rm RF}\mathbf{B}_{n}}\right) + {N_{o}}}},
\end{aligned}
\end{equation}
where the quantities $\mathbf{H}_{m}\in \mathbb{C}^{N_\mathrm{t} \times N_\mathrm{t}}$ and $\mathbf{B}_{m}\in \mathbb{C}^{N_\mathrm{RF}\times N_\mathrm{RF}}$ are given by $\mathbf{H}_{m}=\mathbf {h}_{m}\mathbf {h}_{m}^{H} $ and $\mathbf{B}_{m}=\mathbf{f}_{{\rm BB},m}\mathbf{f}_{{\rm BB},m}^H$, respectively. Consequently, one can reformulate the SINR constraint (\ref{cons:SINR_1}) as follows
\begin{subequations}
\begin{align}
&\frac{{\mathrm{tr}\left({\mathbf{F}_{\rm RF}^H\mathbf{H}_{m}\mathbf{F}_{\rm RF}\mathbf{B}_{m}}\right)}}{{\sum \limits_{\begin{subarray}{l} n = 1 \\ n \ne m \end{subarray}}^{M} \mathrm{tr}\left({\mathbf{F}_{\rm RF}^H\mathbf{H}_{m}\mathbf{F}_{\rm RF}\mathbf{B}_{n}}\right) + {N_{o}}}} \geq {\Gamma_{m}},\forall m, \\
&\Rightarrow{\sum \limits_{\begin{subarray}{l} n = 1 \\ n \ne m \end{subarray}}^{M}  \mathrm{tr}(\mathbf{F}_{\rm RF}^H\mathbf{H}_{m}\mathbf{F}_{\rm RF}\mathbf{B}_{n})}-\frac{1}{\Gamma_{m}}{\mathrm{tr}\left({\mathbf{F}_{\rm RF}^H\mathbf{H}_{m}\mathbf{F}_{\rm RF}\mathbf{B}_{m}}\right)} \nonumber\\
& \qquad\qquad\qquad\qquad\qquad\qquad+ {N_{o}} \leq 0 , \forall m.\label{cons:SINR_3}
\end{align}
\end{subequations}
Upon employing the $\mathrm{vec}(\cdot)$ operation to (\ref{cons:SINR_3}), one can rewrite it as follows
\begin{equation}
\begin{aligned} 
&\sum_{\substack{n = 1 \\ n \ne m}}^{M} \mathrm{vec}(\mathbf{F}_{\mathrm{RF}})^{H} \left(\mathbf{B}_n^{T} \otimes \mathbf{H}_m\right) \mathrm{vec}(\mathbf{F}_{\mathrm{RF}}) \\
&-\frac{1}{\Gamma_{m}} \mathrm{vec}(\mathbf{F}_{\mathrm{RF}})^{H} \left(\mathbf{B}_m^{T} \otimes \mathbf{H}_m\right) \mathrm{vec}(\mathbf{F}_{\mathrm{RF}}) + N_{\mathrm{o}} \leq 0. \label{cons:SINR_4}
\end{aligned}
\end{equation}
Furthermore, we define $\mathbf{d} = \mathrm {vec} (\mathbf {F}_{\rm RF}) \in {\mathbb C}^{{N_\mathrm{t}N_\mathrm{RF}} \times {1}}$, where $\mathbf{d}(l)=1, \forall l$ and $\boldsymbol{\Upsilon}_{n,m}=\mathbf{B}^T_n \otimes \mathbf{H}_m\in \mathbb{C}^{N_\mathrm{t}N_\mathrm{RF} \times N_\mathrm{t}N_\mathrm{RF}}$. Then (\ref{cons:SINR_4}) is equivalently written as 
\begin{equation}
g_m(\mathbf{d}) \triangleq \mathbf{d}^{H} \boldsymbol{\Delta}_{m}\mathbf{d}+ {N_{o}} \leq 0,
\end{equation}
 where $\boldsymbol{\Delta}_{m} = \left ({\sum \limits_{\begin{subarray}{l} n = 1 \\ n \ne m \end{subarray}}^{M}  \boldsymbol{\Upsilon }_{n,m} -\frac{1}{\Gamma_{m}}\boldsymbol{\Upsilon}_{m,m}}\right)\in \mathbb{C}^{N_\mathrm{t}N_\mathrm{RF} \times N_\mathrm{t}N_\mathrm{RF}}$. 
In a similar fashion, let us define  $\mathbf{T} = \mathbf{F}_\mathrm{BB}\mathbf{F}_\mathrm{BB}^H \in \mathbb{C}^{N_\mathrm{RF} \times N_\mathrm{RF}}$. Consequently, we express the total power constraint  \text{(\ref{cons:TP})} as 
\begin{subequations}\label{cons:TP_1}
\begin{align} 
 ||\mathbf{F}_\mathrm{RF}\mathbf{F}_\mathrm{BB}\|_F^2 = &\mathrm{tr}\left({\mathbf{F}_\mathrm{BB}^H\mathbf{F}_{\rm RF}^H\mathbf{F}_{\rm RF}\mathbf{F}_\mathrm{BB}}\right)\\
 = & \mathrm{tr}\left({\mathbf{F}_{\rm RF}^H\mathbf{I}_\mathrm{N_{t}}\mathbf{F}_{\rm RF}\mathbf{T}}\right).
  \end{align}
\end{subequations}
Subsequently, by employing $\mathbf{d} = \mathrm {vec} (\mathbf {F}_{\rm RF})$ in the above equation, (\ref{cons:TP}) can be redefined as 
\begin{equation}
\omega_\mathrm{p}(\mathbf{d}) \triangleq \mathbf{d}^{H} \boldsymbol{\Omega}_\mathrm{p}\mathbf{d} - P_\mathrm{max} \leq 0,
\end{equation}
where $\boldsymbol{\Omega}_\mathrm{p} = \left ({\mathbf {T}^{T} \otimes \mathbf{I}_\mathrm{N_{t}}}\right) \in \mathbb{C}^{\mathbb{C}^{N_\mathrm{t}N_\mathrm{RF} \times N_\mathrm{t}N_\mathrm{RF}}}$.
In addition, we express the objective function RBE of (\ref{RF_1}) using the $\mathrm{vec}(\cdot)$ operation as follows 
\begin{subequations}\label{eqn:obj_1}
\begin{align}
 ||\mathbf{F}_\mathrm{RF}\mathbf{F}_\mathrm{BB}-\mathbf{F}_\mathrm{r}\mathbf{U}\|_F^2 = &||\mathrm {vec}\left( \mathbf{F}_\mathrm{RF}\mathbf{F}_\mathrm{BB}\right) -\mathrm {vec}\left(\mathbf{F}_\mathrm{r}\mathbf{U}\right)\|_2^2\\
 =& || \left ({\mathbf {\mathbf{F}^{T} _\mathrm{BB}}\otimes \mathbf{I}_{N_\mathrm{t}}}\right)\mathbf{d} -\mathrm {vec}\left(\mathbf{F}_\mathrm{r}\mathbf{U}\right)\|_2^2.
  \end{align}
\end{subequations}
Subsequently, the RBE can be expressed in terms of $\mathbf{d}$ as follows
\begin{equation}
    \omega_\mathrm{r}(\mathbf{d})=||\boldsymbol{\Omega}_{r} \mathbf{d} -\mathbf{f}_\mathrm{r}\|_2^2, \label{objec_1}
\end{equation}
where the quantities $\mathbf{f}_\mathrm{r} \in \mathbb{C}^{N_\mathrm{t}M \times 1}$ and $\boldsymbol{\Omega}_{r} \in \mathbb{C}^{N_\mathrm{t}M \times N_\mathrm{t}N_\mathrm{RF}}$ are defined as $\mathbf{f}_\mathrm{r} = \mathrm {vec}\left(\mathbf{F}_\mathrm{r}\mathbf{U}\right)$ and $\boldsymbol{\Omega}_{r} = \left(\mathbf{F}^T_\mathrm{BB}\otimes \mathbf{I}_{N_\mathrm{t}}\right)$, respectively. 
As a further advance, (\ref{objec_1}) is translated into a quadratic expression as follows
\begin{subequations}\label{eqn:obj_3}
\begin{align}
 \omega_\mathrm{r}(\mathbf {d})\triangleq & ||\boldsymbol{\Omega}_\mathrm{r} \mathbf{d} -\mathbf{f}_\mathrm{r}\|_2^2\\
 = & (\boldsymbol{\Omega}_\mathrm{r} \mathbf{d} -\mathbf{f}_\mathrm{r})^H(\boldsymbol{\Omega}_\mathrm{r} \mathbf{d} -\mathbf{f}_\mathrm{r})   \\
 =& \mathbf{d}^{H}  \boldsymbol{\Xi}_\mathrm{r}\mathbf{d} - 2{{\mathrm{Re}}\left(\mathbf{a}^{H}_{r}\mathbf{d}\right)} + e_r,
 \end{align}
\end{subequations}
where $\boldsymbol{\Xi}_\mathrm{r} = \boldsymbol{\Omega}^{H}_\mathrm{r}\boldsymbol{\Omega}_\mathrm{r} \in \mathbb{C}^{N_\mathrm{t}N_\mathrm{RF} \times N_\mathrm{t}N_\mathrm{RF}}$, $\mathbf{a}_{r} = \boldsymbol{\Omega}^{H}_\mathrm{r}\mathbf{f}_\mathrm{r}\in \mathbb{C}^{N_\mathrm{t}N_\mathrm{RF} \times 1}$ and $e_\mathrm{r} = \mathbf{f}^{H}_\mathrm{r}\mathbf{f}_\mathrm{r}$. 
Therefore, following the above mathematical manipulations spanning from (\ref{cons:SINR_7}) to (\ref{eqn:obj_3}), the OP (\ref{RF_1}) can be recast as follows 
\begin{subequations}\label{RF_2}
\begin{align} 
&\min \limits_{\mathbf{d}}~ \omega_\mathrm{r}(\mathbf{d}) ,\label{eqn:quad_obj} \\ 
&\mathrm{s. t.}\quad g_{m}(\mathbf {d}) \leq 0  ,\forall m,\label{cons:quad_con} \\
&\quad \quad \omega_\mathrm{p}(\mathbf {d}) \leq 0,\label{cons:quad_power}  \\
&\quad \quad \left\vert\mathbf{d}(l)\right\vert = 1, \forall  l.  \label{eqn:UMC}
\end{align}
\end{subequations}
The above problem (\ref{RF_2}) is a quadratically constrained quadratic program (QCQP) with an extra UM constraint, which is non-convex. To solve this problem, we propose two innovative methods, namely the bisection-based majorization-minimization (BMM) method and the exact penalty manifold optimization (EPMO), which are discussed next.
\subsubsection*{\textbf{BMM optimization}}
Let us assume $\mathbf{d}^{(\kappa)}$ to be the feasible point for the problem (\ref{RF_2}) that is found from the $(\kappa - 1)$th iteration. Following the inequalities (\ref{eq:MM_3}), (\ref{eq:MM_4}) of Appendix B of \cite{}, the respective majorizer functions for (\ref{eqn:quad_obj}), (\ref{cons:quad_con}) and (\ref{cons:quad_power}) are given by (\ref{MM:RBE}), (\ref{MM:SINR}) and (\ref{MM:TP}), respectively, where $D = N_\mathrm{t}N_\mathrm{RF}$ and  $c_\mathrm{r} = \mathrm{tr}\left(\boldsymbol{\Xi}_\mathrm{r}\right)$ , $c_{m} = \mathrm{tr}\left(\sum \limits_{\begin{subarray}{l}  n = 1 \\ n \ne m \end{subarray}}^{M}  \boldsymbol{\Upsilon}_{n,m}\right)$, and $c_{p} = \mathrm{tr}\left(\boldsymbol {\Omega}_\mathrm{p}\right)$.
\begin{figure*}[t]
\begin{align} 
\omega_\mathrm{r}(\mathbf {d}) &\leq \omega_\mathrm{r}^{(\kappa)}(\mathbf {d}) \triangleq\!2\mathrm{Re}\left(\mathbf{d}^{H}\big[(\boldsymbol{\Xi}_\mathrm{r} \!-\! c_{r}\mathbf {I}_D) \mathbf {d}^{(\kappa)} \!-\! \mathbf {a}_{r} \big]\right) \!-(\mathbf{d}^{(\kappa)})^{H} \boldsymbol{\Xi}_\mathrm{r}\mathbf {d}^{(\kappa)} +\!2Dc_\mathrm{r}\!+\! E_\mathrm{r}, \label{MM:RBE}\\
g_{m}(\mathbf {d}) &\leq {g}^{(\kappa)}_{m}(\mathbf {d}) \triangleq\!2\mathrm{{Re}}\left(\mathbf {d}^{H}(\mathbf {\Delta}_{m}-c_{m}\mathbf{I}_D)\mathbf {d}^{(\kappa)}\right) \! -(\mathbf{d}^{(\kappa)})^{H} \mathbf {\Delta}_{m}\mathbf {d}^{(\kappa)} +\!2Dc_m + {N_{o}} ,\forall m, \label{MM:SINR}\\
\omega_\mathrm{p}(\mathbf {d}) &\leq {\omega}^{(\kappa)}_\mathrm{p}(\mathbf {d}) \triangleq\!2\mathrm{{Re}}\left(\mathbf {d}^{H}(\boldsymbol{\Omega}_\mathrm{p}-c_\mathrm{p}\mathbf{I}_D)\mathbf {d}^{(\kappa)}\right) \! -(\mathbf{d}^{(\kappa)})^{H} \boldsymbol {\Omega}_\mathrm{p}\mathbf{d}^{(\kappa)} +\!2Dc_\mathrm{p} - P_\mathrm{max},\label{MM:TP}
\end{align}
\normalsize
\hrulefill
\end{figure*}
Thus, to generate the next feasible point $\mathbf {d}^{(\kappa+1)}$, we solve the following MM OP in the $\kappa$th iteration
\begin{subequations}\label{RF_3}
\begin{align} 
&\min \limits_{\mathbf{d}}~ \omega^{(\kappa)}_\mathrm{r}(\mathbf{d}) \\ 
\mathrm{s. t.} &\quad {g}^{(\kappa)}_{m}(\mathbf {d}) \leq 0  ,\forall m, \\
&\quad \omega^{(\kappa)}_\mathrm{p}(\mathbf {d}) \leq 0, \\
&\quad \vert \mathbf{d}(l)\vert = 1, \forall l.
\end{align}
\end{subequations}
To solve the OP (\ref{RF_3}), we adopt the Lagrange dual optimization by applying the Karush-Kuhn-Tucker (KKT) framework \cite{boyd2004convex}. 
To this end, the Lagrange function associated with the problem (\ref{RF_3}) is given by
\begin{equation} \label{eqn:LF}
\mathcal{L}^{(\kappa)}(\mathbf {d},\boldsymbol{\lambda}^{(\kappa)}, \vartheta^{(\kappa)}) = \omega^{(\kappa)}_\mathrm{r}(\mathbf {d}) + \sum _{m=1}^{M} \lambda_{m}^{(\kappa)} {g^{(\kappa)}_{m}}(\mathbf {d}) + \vartheta^{(\kappa)} {\omega^{(\kappa)}_\mathrm{p}}(\mathbf {d}),
\end{equation}
where $\boldsymbol{\lambda}^{(\kappa)}  =[\lambda^{(\kappa)}_1,\hdots,\lambda^{(\kappa)}_M]^T \in {\mathbb C}^{M \times 1}$. Moreover, $\lambda^{(\kappa)}_m \geq 0$ and $\vartheta^{(\kappa)}\geq 0$ denote the Lagrange dual multiplier associated with $g^{(\kappa)}_m$ and $\omega^{(\kappa)}_\mathrm{p}$, respectively, in the $\kappa$th iteration. Subsequently, the Lagrange dual function $\mathcal{D}^{(\kappa)}_{\mathbf{d}}(\mathbf{\lambda}^{(\kappa)}, \vartheta^{(\kappa)})$ over the variable $\mathbf {d}$ is expressed as
\begin{equation} \label{eqn:LD_1}
\mathcal{D}^{(\kappa)}_{\mathbf {d}}(\mathbf {\lambda}^{(\kappa)}, \vartheta^{(\kappa)}) = \min _{ \left\vert\mathbf{d}(l)\right\vert = 1, \forall  l \in \mathcal{L} }  \mathcal{L}^{(\kappa)}(\mathbf {d},\mathbf {\lambda}^{(\kappa)}, \vartheta^{(\kappa)}).  
\end{equation}
Since the Lagrange function $\mathcal{L}^{(\kappa)}(\mathbf {d},\boldsymbol{\lambda}^{(\kappa)}, \vartheta^{(\kappa)})$ is linear with respect to the variable $\mathbf{d}$, the primal optimal point of (\ref{eqn:LD_1}) can be written as a function of the dual multipliers as follows
\begin{equation}\label{eqn:PO} 
\begin{aligned}
&\mathbf {d}^{(\kappa)}(\boldsymbol{\lambda}^{(\kappa)}, \vartheta^{(\kappa)}) = \exp \Bigg( j\hspace{2mm} \mathrm{arg} \Big[ (c_\mathrm{r}\mathbf{I}_D-\boldsymbol{\Xi}_\mathrm{r})\mathbf {d}^{(\kappa)}\!+\! \mathbf {a}_\mathrm{r}  \\
&+ \sum _{m=1}^{M} \big[(c_{m}\mathbf {I_{D}}-\mathbf {\Delta}_{m})\mathbf {d}^{(\kappa)}\big] \lambda _{m}^{(\kappa)}+ \big[(c_\mathrm{p}\mathbf {I_{D}}-\boldsymbol{\Omega}_\mathrm{p})\mathbf {d}^{(\kappa)}\big] \vartheta^{(\kappa)} \Big] \Bigg).   
\end{aligned}
\end{equation}
Furthermore, the optimal solution for the Lagrange Dual problem (\ref{eqn:LD_1}) is obtained as follows
\begin{equation} \label{eqn:DP}
\begin{aligned}
 &\{\boldsymbol{\lambda}^{(\kappa +1)}, \vartheta^{(\kappa +1)}\} \\
 &=  \arg \hspace{-3mm}\sup_{ \lbrace \lambda^{(\kappa)}_{m}\geq 0\rbrace _{m=1}^{M}, \vartheta^{(\kappa)}\geq 0} \hspace{-3mm}\mathcal{D}_{\mathbf{d}}^{(\kappa)}(\mathbf{\lambda}^{(\kappa)}, \vartheta^{(\kappa)}) \vert_{\mathbf {d} = \mathbf {d}^{(\kappa)}(\boldsymbol{\lambda}^{(\kappa)}, \vartheta^{(\kappa)})}. 
 \end{aligned}
 \end{equation}
Due to the strong duality between the primal problem (\ref{RF_3}) and its dual (\ref{eqn:DP}), the primal optimal point $\mathbf {d}^{(\kappa +1)}$ can be obtained by employing the KKT conditions as follows
\begin{subequations}\label{KKT}
\begin{align}
&\mathbf{d}^{(\kappa+1)}\left(\boldsymbol{\lambda}^{(\kappa)}, \vartheta^{(\kappa)}\right)=  \exp \Bigg( j\hspace{2mm} \mathrm{arg} \Big[ (c_\mathrm{r}\mathbf{I}_D-\boldsymbol{\Xi}_\mathrm{r})\mathbf {d}^{(\kappa)}\!+\! \mathbf {a}_\mathrm{r} \nonumber \\
&+ \sum _{m=1}^{M} \big[(c_{m}\mathbf {I_{D}}-\mathbf {\Delta}_{m})\mathbf {d}^{(\kappa)}\big] \lambda _{m}^{(\kappa)}+ \big[(c_\mathrm{p}\mathbf {I_{D}}-\boldsymbol{\Omega}_\mathrm{p})\mathbf {d}^{(\kappa)}\big] \vartheta^{(\kappa)} \Big] \Bigg), \\
&\quad 0\leq \lambda_{m}^{(\kappa)} \leq \infty, \quad g_m^{(\kappa)}\left(\mathbf{d}^{\left(\kappa\right)}\left(\boldsymbol{\lambda}^{(\kappa)}\right)\right)\leq 0 \label{eqn:DPF_SINRS}, \quad \forall m\\
&\quad 0\leq \vartheta^{(\kappa)} \leq \infty, \quad \omega_\mathrm{p}^{(\kappa)}\left(\mathbf{d}^{(\kappa)}\left(\vartheta^{(\kappa)}\right)\right)\leq 0 \label{eqn:DPF_POWER}, \\
&\quad \lambda_{m}^{(\kappa)} g_m^{(\kappa)}\left(\mathbf{d}^{(\kappa)}\left(\boldsymbol{\lambda}^{(\kappa)}\right)\right) = 0 \quad \forall m \label{eqn:CS_SINRS},\\
&\quad \vartheta^{(\kappa)} \omega_\mathrm{p}^{(\kappa)}\left(\mathbf{d}^{(\kappa)}\left(\vartheta^{(\kappa)}\right)\right) = 0, \label{eqn:CS_POWER}
\end{align}
\end{subequations}
where (\ref{eqn:DPF_SINRS}) and (\ref{eqn:DPF_POWER}) are the dual and primal feasibility for ${g^{(\kappa)}_{m}}$ and $\omega^{(\kappa)}_\mathrm{p}$, respectively. Furthemore, (\ref{eqn:CS_SINRS}) and (\ref{eqn:CS_POWER}) indicate the complementary slackness of the functions ${g^{(\kappa)}_m}$ and $\omega^{(\kappa)}_\mathrm{p}$, respectively. 
Therefore, one can compute the next feasible point of the dual multipliers $\boldsymbol{\lambda}^{(\kappa +1)}$ and $\vartheta^{(\kappa+1)}$ by solving (\ref{eqn:DP}) with the aid of the KKT conditions (\ref{KKT}) for the given feasible primal point $\mathbf{d}^{(\kappa)}$ and dual multipliers $\boldsymbol{\lambda}^{(\kappa)}, \vartheta^{(\kappa)}$ in the $\kappa$th iteration. We adopt the coordinate ascent technique to compute the next feasible point of the dual multipliers $\boldsymbol{\lambda}^{(\kappa +1)}$ and $\vartheta^{(\kappa+1)}$.
In addition, for the given $\{ \lambda^{(\kappa)} _{n} \geq 0 \}_{n=1,n \neq m }^{M}$ and $\vartheta^{(\kappa)}$, if ${g}^{(\kappa)}_{m}\left(\mathbf{d}^{\left(\kappa\right)}\left(\boldsymbol{\lambda}^{\left(\kappa\right)}\right)\right)\vert_{\lambda^{(\kappa)}_{m} = 0} \leq 0$, we set $\lambda^{(\kappa+1)} _{m} = 0$ by following (\ref{eqn:CS_SINRS}). Otherwise, there exists a non-zero ${g}^{(\kappa)}_{m}\left(\mathbf{d}^{\left(\kappa\right)}\left(\boldsymbol{\lambda}^{\left(\kappa\right)}\right)\right)\vert{\lambda^{(\kappa+1)}_m} \approx 0$. Similarly, for a given $ \vartheta^{(\kappa)} \geq 0$, if ${\vartheta}^{(\kappa)}\omega_\mathrm{p}\left(\mathbf {d}^{(\kappa)}\left(\vartheta^{(\kappa)}\right)\right)\vert_{\vartheta^{(\kappa)} = 0} \leq 0$, then we set $\vartheta^{(\kappa+1)} = 0$ by following (\ref{eqn:CS_POWER}). Otherwise, there exists a non-zero ${\vartheta}^{(\kappa)}\omega_\mathrm{p}\left(\mathbf {d}^{(\kappa)}\left(\vartheta^{(\kappa)}\right)\right)\vert{\vartheta^{(\kappa+1)}} \approx 0$. 
To obtain such non-zero dual multipliers within a limited number of iterations, we employ the bisection method. 
Finally, if ${ \lbrace \lambda^{(\kappa)}_{m}\rbrace _{m=1 }^{M}}$ and $\vartheta^{(\kappa)}$ satisfy all the constraints, then the next feasible point of the primal OP is found as 
\begin{equation} 
\mathbf {d}^{(\kappa +1)} = \exp \left(\mathrm{arg} \hspace{2mm} \left[\left(c_\mathrm{r}\mathbf{I}_D-\mathbf{\Xi}_\mathrm{r}\right)\mathbf{d}^{(\kappa)}\!+\! \mathbf{a}_\mathrm{r}\right]\right).
\end{equation}  
\begin{algorithm}[t]
\caption{BMM algorithm for solving (\ref{RF_2})}
 \textbf{Input:} Feasible $\mathbf{F}_\mathrm{RF}$, $\mathbf{F}_\mathrm{BB}$, and stopping parameters $\tau_1$, $\tau_2$\\
 \textbf{Output:} Optimal RF TPC $\mathbf{F}^*_\mathrm{RF}$
 \label{alg:BMM}
\begin{algorithmic}[1]
\State \textbf{Initialize:} $\kappa = 0$, $\nu = 0$, $\mathbf{d}^{(\kappa)} = \operatorname{vec}(\mathbf{F}_\mathrm{RF})$, and dual multipliers $\{\lambda^{(\kappa)}_{m} \}_{m=1}^{M} = 0$ and $\vartheta^{(\kappa)} = 0$
\State \textbf{ repeat}
\State \hspace{0.25cm} \textbf{for} $m = 1$ to $K$
     \State \hspace{0.5cm} $\textbf{if}$ ${g}^{(\kappa)}_{m}\left(\mathbf{d}^{\left(\kappa\right)}\left(\boldsymbol{\lambda}^{\left(\kappa\right)}\right)\right)\vert_{\lambda^{(\kappa)}_{m} = 0} \leq 0$, $\textbf{then}$ ${\lambda_{m} = 0}$ 
    \State \hspace{0.5cm} $\textbf{else}$
    \State \hspace{1cm} $\lambda_{m}^{L} = 0 $  and ${\lambda_{m}^{U}= 1}$
    \State \hspace{1cm} $\textbf{if}$ ${g}^{(\kappa)}_{m}\left(\mathbf{d}^{\left(\kappa\right)}\left(\boldsymbol{\lambda}^{\left(\kappa\right)}\right)\right)\vert _{\lambda^{(\kappa)}_{m} = 1} \leq 0$ , $\textbf{then}$ ${\lambda_{m}^{U} = 1}$ 
    \State \hspace{1cm} $\textbf{else}$ 
    \State \hspace{1.5cm} $\textbf{while }$  {${g}^{(\kappa)}_{m}\left(\mathbf{d}^{\left(\kappa\right)}\left(\boldsymbol{\lambda}^{\left(\kappa\right)}\right)\right)\vert_{\lambda^{(\kappa)}_{m} = \lambda^{U}_{m}} \geq 0$} $\textbf{do}$
    \State \hspace{2.5cm} $\lambda^{U} _{m} = 2\lambda^{U} _{m}$
    \State \hspace{1.5cm} $\textbf{end while}$
    \State \hspace{1.5cm} $\lambda^{L} _{m} = \lambda^{U} _{m}/{2}$
    \State \hspace{1cm} $\textbf{end}$ $\textbf{if}$
     \State \hspace{1cm} $ \lambda_{m} = (\lambda^{L}_m + \lambda^{U}_m) / 2$ 
     \State \hspace{1cm} $\textbf{while }$ $|{g}^{(\kappa)}_{m}\left(\mathbf{d}^{\left(\kappa\right)}\left(\boldsymbol{\lambda}^{\left(\kappa\right)}\right)\right)\vert _{\lambda^{(\kappa)}_{m} = \lambda _{m}} + \frac{\tau_1}{2}| \geq  \frac{\tau_1}{2}$ $\textbf{do}$
     \State \hspace{1.25cm} $\textbf{if}$ ${g}^{(\kappa)}_{m}\left(\mathbf{d}^{\left(\kappa\right)}\left(\boldsymbol{\lambda}^{\left(\kappa\right)}\right)\right)\vert _{\lambda^{(\kappa)}_{m}=\lambda _{m}}\geq0 $, $\textbf{set}$  $\lambda^{L}_m=\lambda _{m}$ 
    \State \hspace{1.25cm} $\textbf{else}$ $\lambda^{U}_m  = \lambda _{m}$ 
    \State \hspace{1.25cm} $\textbf{end}$ $\textbf{if}$
    \State \hspace{1cm} $\textbf{end while}$
    \State \hspace{0.5cm} $\textbf{end}$ $\textbf{if}$
    \State \hspace{0.25cm} \textbf{end} \textbf{for}
    \State \hspace{0.25cm} follow steps from (4) to (20) to get $\vartheta$ using function ${\omega}^{(\kappa)}_\mathrm{p}\left(\mathbf{d}^{(\kappa)}\left(\boldsymbol{\lambda}^{(\kappa)}\right)\right)$ for given $\{ \lambda_{m} \}_{m=1}^{M}$

     \State \hspace{0.25cm} $\textbf{set}$ $\nu \leftarrow \nu+1$, $\boldsymbol{\lambda }^{(\kappa)}[\nu]$ = $[{\lambda_1},\hdots,{\lambda}_M]^T$ and $\vartheta^{\kappa}[\nu] = \vartheta$, 
     \State \hspace{0.25cm} compute $L[\nu] = \mathcal{L}^{(\kappa)}\left(\mathbf {d},\boldsymbol {\lambda}^{(\kappa)}, \vartheta^{(\kappa)}\right)\big\vert_{\mathbf {d} = \mathbf{d}^{(\kappa)}\left(\boldsymbol{\lambda}^{(\kappa)}[\nu], \vartheta^{(\kappa)}[\nu]\right)}$
    \State $\textbf{until}$ $\left\vert\left({L[\nu]-L[\nu - 1]}\right)/L[\nu]\right\vert \leq \tau_1 $
    \State $\textbf{set}$ $\kappa \leftarrow \kappa+1$ and  $\mathbf {d}^{(\kappa + 1)} = \mathbf {d}^{(\kappa)}\left(\boldsymbol{\lambda}^{(\kappa)}[\nu], \vartheta^{(\kappa)}[\nu]\right) $
    \State evaluate $ \omega_\mathrm{r}(\mathbf {d})$ using (\ref{eqn:obj_3}) and set $\omega_\mathrm{r}^{(\kappa)} =  \omega_\mathrm{r}\left(\mathbf{d}^{(\kappa + 1)}\right)$
    \State $\textbf{if}$ $\left\vert(\omega_\mathrm{r}^{(\kappa)}- \omega_\mathrm{r}^{(\kappa-1)})/\omega_\mathrm{r}^{(\kappa)}\right\vert \leq \tau_2 $, \textbf{then} $\mathbf {d^{*}} = \mathbf {d}^{(\kappa + 1)} $ \textbf{stop}
    \State $\textbf{else}$ \textbf{go to step 2}
    \State $\textbf{end}$ $\textbf{if}$
    \State \textbf{return:} $\mathbf{F}^*_\mathrm{RF}=\text{reshape}(\mathbf {d^{*}})$ to ${N_\mathrm{t}} \times {N_\mathrm{RF}}$ matrix.
\end{algorithmic}
\end{algorithm}
Algorithm \ref{alg:BMM} summarizes the computational procedure of the BMM method harnessed for solving the problem (\ref{RF_2}).
Note that the next feasible point for the dual multipliers $\{\lambda^{(\kappa +1 )}_m, \vartheta^{(\kappa+1)}\}$ of the problem (\ref{eqn:DP}) is obtained by optimizing a single dual multiplier at a time, while keeping the other dual multipliers fixed, until the Lagrange function (\ref{eqn:LF}) converges at the corresponding primal optimal point $\mathbf {d}^{(\kappa +1)}$. 
Moreover, observe that the following conditions hold
\begin{equation} \label{IMM_C}
{\omega}_\mathrm{r}\left(\mathbf{d}^{(\kappa +1 )}\right) \leq {\omega}^{(\kappa)}_\mathrm{r}\left(\mathbf {d}^{(\kappa +1 )}\right)  < \omega^{(\kappa)}_\mathrm{r}\left(\mathbf {d}^{(\kappa)}\right) = \omega_\mathrm{r}\left(\mathbf {d}^{(\kappa)}\right),  
\end{equation}    
which shows that $\mathbf{d}^{(\kappa+1)}$ is an improved feasible point over $\mathbf{d}^{(\kappa)}$ for problem (\ref{RF_2}). Thus, Algorithm \ref{alg:BMM} generates a sequence of upgraded feasibile points by iteratively solving the problem (\ref{RF_2}) until it converges to its locally optimal solution.

\subsubsection*{ \textbf{EPMO method}}
The BMM algorithm requires a bisection search for $M+1$ dual multipliers, which renders it highly complex. Therefore, we propose a low-complexity exact penalty-based manifold optimization (EPMO) method to solve (\ref{RF_2}). To relax the SINR constraint (\ref{cons:quad_con}) and TPC constraint (\ref{cons:quad_power}), we add them into the objective function as a penalty term and subsequently, the problem is solved by employing manifold optimization.
To this end, let us redefine the SINR constraint (\ref{cons:quad_con}) as $\psi_m\left(\mathbf {d}\right) \triangleq \left(\max\left(0,g_m\left(\mathbf{d}\right)\right)\right)^2,\forall m$ and (\ref{cons:quad_power}) as $\chi_\mathrm{p}(\mathbf{d})\triangleq \left(\max\left(0, \omega_\mathrm{p}\left(\mathbf{d}\right)\right)\right)^2$. Consequently, the equivalent OP for (\ref{RF_2}) can be expressed as
\begin{equation}\label{EPMO}
\begin{aligned}
& \min_{\mathbf{d}}~  f(\mathbf {d}) = \omega_\mathrm{r}\left(\mathbf{d}\right) + \mu \left(\sum_{\begin{subarray}{l} m = 1 \\  \end{subarray}}^{M} \psi_m\left(\mathbf{d}\right)+\chi_\mathrm{p}\left(\mathbf{d}\right)\right)\\ 
&\text{s.t.} \quad (\ref{eqn:UMC}),
\end{aligned}
\end{equation}
where $\mu > 1$ is a penalty factor. Specifically, $\mu$ is obtained by adopting sequential optimization, wherein the penalty parameter $\mathbf{\mu}$ is increased successively, followed by solving the problem (\ref{EPMO}) until the solutions eventually converge to that of the original problem (\ref{RF_2}).
Observe that the constraint (\ref{eqn:UMC}) represents a complex circle Riemannian manifold  $\mathcal{M} = \{\mathbf{d} \in {\mathbb C}^{{N_\mathrm{t}}{N_\mathrm{RF}} \times 1}:  \left|\mathbf{d}(l)\right| = 1,\forall   1 \leq l \leq N_\mathrm{t}N_\mathrm{RF} \}$. Therefore, (\ref{EPMO}) can be solved by using the manifold optimization method. Specifically, for $\mathbf{\mu} > 1$, we adopt the Riemannian conjugate gradient (RCG) optimization method to find a near-optimal solution. Note that the RCG algorithm relies on computing the Riemannian gradient to obtain the steepest direction in the decreasing objective function. However, computing the Riemannian gradient differs from obtaining the traditional gradient in the Euclidean space. Toward this, let us evaluate the Euclidean gradient of the objective function $f(\mathbf {d})$ as follows
\begin{equation}  
\nabla f(\mathbf {d}) =  2\boldsymbol{\Xi}_\mathrm{r}\mathbf {d}-2\mathbf{a}_\mathrm{r} + \mathbf{\mu} \left(\sum_{m = 1}^{M} \boldsymbol{\xi}_m + \boldsymbol{\xi}_\mathrm{p}\right),  
\end{equation}
where $\boldsymbol{\xi}_m$ and $\boldsymbol{\xi}_p$ are given by
\begin{subequations}\label{eqn:EG} 
\begin{align}
&\boldsymbol{\xi}_m = \begin{cases}
     4g_{m}(\mathbf{d})\boldsymbol{\Delta}_{m}\mathbf{d},& \text{if } g_{m}(\mathbf{d}) \geq 0, \\
    0, & \text{otherwise},
\end{cases}\\
&\boldsymbol{\xi}_p = \begin{cases}
     4\omega_\mathrm{p}(\mathbf{d})\boldsymbol{\Omega}_\mathrm{p}\mathbf{d},& \text{if } \omega_\mathrm{p}(\mathbf{d}) \geq 0 \\
    0, & \text{otherwise}.
\end{cases}
\end{align}
\end{subequations}
Furthermore, computing the Riemannian gradient involves the tangent space, which comprises the vectors that are tangential to any smooth curves on the manifold $\mathcal{M}$. In addition, the tangent space at a point $\mathbf{d}$ on the complex circle manifold $\mathcal{M}$ is defined as
\begin{equation}
 T_{\mathbf {d}} \mathcal {M} = \lbrace \mathbf {z}\in {{\mathbb {C}}^{N_{t}N_{RF}\times 1}}|\mathrm{Re} \left(\mathbf {z} \odot \mathbf {d}^*\right) = \mathbf {0}_{N_{t}N_{RF}\times 1}\rbrace.  
 \end{equation}
Thereby, the Riemannian gradient $\nabla _{\mathcal {M}} f({\mathbf {d}})$ can be obtained by projecting $\nabla f(\mathbf{d})$ onto the tangent space of the manifold $\mathcal {M}$ using a projection operator, which is given by
\begin{equation}\label{eqn:RG} 
\begin{aligned}
\nabla _{\mathcal {M}} f({\mathbf {d}}) & = {\text{Pro}}{{\text{j}}_{\mathbf {d}}}\nabla f({\mathbf {d}}) \\ & = \nabla f({\mathbf {d}}) - R\text{e} \{ \nabla f({\mathbf {d}}) \odot {{\mathbf {d}}^{*}}\} \odot \mathbf{d}. 
 \end{aligned}
\end{equation}
Employing the Riemannian gradient, one can follow the same steps as that of the Euclidean space for optimization. Thus, the steepest search direction in the $(\kappa + 1)$th iteration is given by
\begin{equation} \label{SD}
\begin{aligned}
\boldsymbol{\zeta }^{(\kappa+1)} = -\nabla _{\mathcal {M}} f\left(\mathbf {d}^{(\kappa+1)}\right) + \rho\; {T}_{\mathbf {d}^{(\kappa)} \mapsto \mathbf {d}^{(\kappa+1)}} \left(\boldsymbol{\zeta }^{(\kappa)} \right), 
\end{aligned} 
\end{equation}
where $\boldsymbol{\zeta }^{(\kappa)}$ is the search direction at $\mathbf{d}^{(\kappa)}$, 
${\mathbf{\rho}}$ is Polak-Ribiére’s conjugate parameter \cite{mm_ISAC_2} and $\mathcal{T}_{\mathbf{d}^{(\kappa)} \rightarrow \mathbf{d}^{(\kappa+1)}}\left ({\boldsymbol{\zeta }^{(\kappa)}}\right)$ is the transport operation used to map the search direction from its original tangent space to the current tangent space. The transport operation is expressed as
\begin{equation}\label{TV}
\begin{aligned}
{\mathcal{T}_{\mathbf{d}}^{(\kappa)}}\mathcal{M} \to &{\mathcal{T}_{{{\mathbf{d}}^{(\kappa + 1)}}}}\mathcal{M}:  \\ 
&{\mathcal{T}}_{\mathbf {d}^{(\kappa)} \mapsto \mathbf {d}^{(\kappa+1)}} \left(\mathbf {\zeta^{(\kappa)}} \right) =\\
&\mathbf {\zeta^{(\kappa)}} - \mathrm{Re} \left\lbrace \mathbf {\zeta^{(\kappa)}} \odot \left(\mathbf{d}^{\left(\kappa+1\right)}\right)^{*} \right\rbrace \odot \mathbf{d}^{(\kappa+1)}.
\end{aligned}
\end{equation} 
Moreover, in the Euclidean gradient, the next point is computed as $\mathbf{d}^{(\kappa+1)} = \mathbf{d}^{(\kappa)} + \mathbf{ \delta^{(\kappa)} \zeta^{(\kappa)}}$ with $\delta^{(\kappa)}$ as step size, which lies on the tangent space $T_{\mathbf {d}} \mathcal {M}$. Therefore, to project the point to the manifold $\mathcal {M}$, we perform retraction mapping \cite{ISAC_1}, which is given by
\begin{equation}\label{Ret}
\begin{aligned}
& {\text{Ret}}{{\text{r}}_{\mathbf{d}}}:{\mathcal{T}_{\mathbf{d}}}\mathcal{M} \to \mathcal{M}: \\ 
& \mathbf{d^{(\kappa+1)}} = \left[ {\frac{{{{({\mathbf{d^{(\kappa)} + }}\mathbf{ \delta^{(\kappa)}\zeta^{(\kappa)}})_1}}}}{{|{{({\mathbf{d^{(\kappa)} + }}\mathbf{ \delta^{(\kappa)}\zeta^{(\kappa)}})}}|_1}}},\hdots,{\frac{{{{({\mathbf{d^{(\kappa)} + }}\mathbf{ \delta^{(\kappa)}\zeta^{(\kappa)}})_D}}}}{{|{{({\mathbf{d^{(\kappa)} + }}\mathbf{ \delta^{(\kappa)}\zeta^{(\kappa)}})}}|_D}}}\right]^T, 
\end{aligned}
\end{equation} 
where $\delta^{(\kappa)}$ is the step size at the $\kappa$th iteration, which is obtained by Armijo's
backtracking line search algorithm \cite{ISAC_1}. 
Furthermore, Algorithm \ref{alg:EPMO} summarizes the complete procedure of solving (\ref{RF_2})  using the EPMO method, which involves updating the penalty parameter $\mathbf{\mu}$ until the measures of violating the constraints (\ref{cons:quad_con}) and (\ref{cons:quad_power}) satisfy the condition
\begin{equation} 
\omega_\mathrm{r}(\mathbf{d}_1)\vert_\mathbf{\mu_1} \geq \omega_\mathrm{r}(\mathbf {d}_2)\vert_\mathbf{\mu_2}, \hspace{2mm} \forall m,
\end{equation}
where $\mathbf{d}_1$ and $\mathbf{d}_2$ are the optimal solutions of problem (\ref{EPMO}) for given $\mathbf{\mu_1} <\mathbf{\mu_2}$, respectively.
\begin{algorithm}[t]
\caption{EPMO algorithm for solving (\ref{RF_2})}
 \textbf{Input:} $\mathbf{F}_\mathrm{BB}$ and thresholds $\tau_3>0$ ,$\tau_4>0$, $0< c < 1$\\
 \textbf{Output:} Optimal RF TPC $\mathbf{F}^*_\mathrm{RF}$
 \label{alg:EPMO}
\begin{algorithmic}[1]
\State \textbf{Initialize:} $\mathbf{F}_\mathrm{RF}$, $\kappa=0$, $\mathbf{d}^{(\kappa)} =\operatorname{vec}(\mathbf{F}_\mathrm{RF})$, $\zeta^{(\kappa)} = - \nabla_{\mathcal {M}} f({\mathbf {d^{(\kappa)}}}) $
\While{$\left(\|\nabla _{\mathcal {M}} f({\mathbf {d^\kappa}}) \|^2 \geq \tau_3 \right)$}
   \State find Armijo backtracking line search step size $\delta^{(\kappa)}$ 
    \State obtain the next point $\mathbf{d}^{(\kappa+1)}$ using the retraction (\ref{Ret})
    \State compute the Riemannian gradient $\nabla_{\mathcal{M}} f\left({\mathbf{d^{\left(\kappa+1\right)}}}\right)$ using (\ref{eqn:RG}).
    \State evaluate the transport ${T}_{\mathbf{d}^{(\kappa)} \mapsto \mathbf{d}^{(\kappa+1})} \left(\boldsymbol{\zeta^{\left(\kappa\right)}}\right)$ using (\ref{TV})
    \State determine the steepest direction $\boldsymbol{\zeta}^{(\kappa+1)}$ according to (\ref{SD})
    \State set $\kappa\leftarrow \kappa+1$
    \EndWhile
    \State \textbf{end} \textbf{while}
    \State $\textbf{if}$ $\left(\sum_{\begin{subarray}{l} m = 1 \\  \end{subarray}}^{M} \psi_m\left(\mathbf{d}^{(\kappa)}\right)+\chi_\mathrm{p}\left(\mathbf{d}^{(\kappa)}\right)\right) \leq \tau_4 , $ 
    \State \hspace{0.5cm} \textbf{return} $\mathbf {d^{*}} = \mathbf {d^{(\kappa)}} $ \textbf{stop}
    \State $\textbf{else}$
    \State \hspace{0.5cm} \textbf{update} $\mathbf{\mu} = \frac{\mu}{c}$ and \textbf{go to step 2}
    \State $\textbf{end}$ $\textbf{if}$
    \State \textbf{return:} $\mathbf{F}^*_\mathrm{RF}=\text{reshape}(\mathbf {d^{*}})$ to ${N_\mathrm{t}} \times {N_\mathrm{RF}}$ matrix
\end{algorithmic}
\end{algorithm}
\subsubsection{Sub-problem for $\mathbf{F}_\mathrm{BB}$}
For the given $\mathbf{F}_\mathrm{RF}$ and $\mathbf{U}$, the resultant OP for $\mathbf{F}_\mathrm{BB}$ is given by
\begin{equation}\label{BB_1}
\begin{aligned} 
& \min \limits_{ \mathbf{F}_\mathrm{BB}}~\|\mathbf{F}_\mathrm{RF}\mathbf{F}_\mathrm{BB}-\mathbf{F}_\mathrm{r}\mathbf{U}\|_F^2\\ 
&\mathrm{s. t.}\quad \gamma_m \geq \Gamma_m, \forall m\\
&\qquad \|\mathbf{F}_\mathrm{RF}\mathbf{F}_\mathrm{BB}\|_F^2\leq P_\mathrm{max}.
\end{aligned}
\end{equation}
To solve the problem (\ref{BB_1}), we reformulate the non-convex SINR constraint as a second-order cone (SOC) constraint by introducing a common phase shift for $\mathbf{F}_{\rm RF}\mathbf{f}_{{\rm BB},m}$. Thus, the equivalent second-order cone programming (SOCP) problem constructed for (\ref{BB_1}) is given by
\begin{equation}\label{eqn:SOCP}
\begin{aligned}
&  \min \limits _{\mathbf{F}_\mathrm{BB}}~\|\mathbf{F}_\mathrm{RF}\mathbf{F}_\mathrm{BB}-\mathbf{F}_\mathrm{r}\mathbf{U}\|_F^2 \\
&{\mathrm{s. t.}}~\left \|{ \begin{array}{c}  \mathbf{A}^H \mathbf{e}\\ \sqrt{N_o} \end{array} }\right \|_2 \le  \sqrt{1+\frac{1}{\Gamma_{m}}} \mathrm{Re}\left(t_{m, n}\right),\\
&\qquad \|\mathbf{F}_\mathrm{RF}\mathbf{F}_\mathrm{BB}\|_F\leq \sqrt{P_\mathrm{max}},
\end{aligned}
\end{equation}
where $t_{m, n}= \mathbf{h}_m^H \mathbf{F}_{\rm RF}\mathbf{f}_{{\rm BB},n}$, $\mathbf{A}(m, n) = t_{m, n}$, and $\mathbf{e} \in \mathbb{C}^{M \times 1}$ is the vector with one in its $m$th position and zero elsewhere. The OP (\ref{eqn:SOCP}) above is a SOCP convex OP, which can be efficiently solved using a standard convex optimization tool package \cite{grant2014cvx}.






\subsubsection{Sub-Problem for $\mathbf{U}$}
For the given $\mathbf{F}_\mathrm{RF}$ and $\mathbf{F}_\mathrm{BB}$, the OP constructed for $\mathbf{U}$ is given by
\begin{equation}\label{UTP}
\begin{aligned}
\mathop {\min }\limits _{\mathbf{U}}&\quad \|\mathbf{F}_\mathrm{RF}\mathbf{F}_\mathrm{BB}-\mathbf{F}_\mathrm{r}\mathbf{U}\|_F^2\\ \mathrm{{s}}.\mathrm{{t}}.&\quad  {\mathbf{U}}{{\mathbf{U}}^H} = \mathbf{I}_{N_\mathrm{tar}}.  
\end{aligned}
\end{equation}
Problem (\ref{UTP}) is the orthogonal Procrustes problem (OPP) \cite{mm_ISAC_5}, which is the least-squares problem associated with a non-convex UM \cite{mm_ISAC_5}. Interestingly, its optimal solution can be obtained via the singular value decomposition (SVD), which is given by
\begin{equation}
\mathbf{U} = \widetilde{\mathbf{U}}\mathbf{I}_{N_\mathrm{tar}\times {M}}\widetilde{\mathbf{V}}^H,
\end{equation}
where ${\mathbf{I}}_{{N_{tar}} \times {{M}}}$ is constructed by concatenating the $[{N_{tar}} \times  {N_{tar}}]$ identity matrix and the $[N_{{tar}} \times ({M} - N_{{tar}})]$ zero matrix, while the matrices $\widetilde{\mathbf{U}}$ and $\widetilde{\mathbf{V}}$ are derived from the following equation
\begin{equation}\label{SVD}
   \mathrm{SVD}{\left(\mathbf{F}^H_\mathrm{r}\mathbf{F}_\mathrm{RF}\mathbf{F}_\mathrm{BB} \right)} = \widetilde{\mathbf{U}} \mathbf{\Sigma} \widetilde{\mathbf{V}}^H.
\end{equation}

\subsection{Outer layer: update $\{\beta_m\}_{m=1}^M$ and $\mathcal{R}$ }
For the fixed RBE,  $\mathcal{E} = \|\mathbf{F}_\mathrm{RF}\mathbf{F}_\mathrm{BB}-\mathbf{F}_\mathrm{r}\mathbf{U}\|_F^2$, the next step is to optimize the block length $\{\beta_m\}_{m=1}^M$ in the outer layer. The equivalent sub-problem is constructed for addressing the blocklength $\{\beta_m\}_{m=1}^M$ as follows
\begin{equation}
\begin{aligned}
&\mathop {\max }\limits_{\{\beta_m\}_{m=1}^M}  \mathcal{R}\\ 
&\text {s. t.}\quad \text{(\ref{cons:rate}), (\ref{cons:BL}), and (\ref{cons:IN})}.
\end{aligned}
\end{equation}
\begin{prop}\label{Prop_2}
To find a point on the Pareto boundary for the given RBE $\mathcal{E}$, the block length constraint must be met with equality, i.e.,
\begin{equation}\label{sum:N}
    \sum_{m=1}^{M} \beta_{m} = N.
\end{equation}
\end{prop}
\begin{proof}
This is proved by considering two CUs in the system. For the given BB and RF TPC, the corresponding rates $R_{\rm 1}$ and $R_{\rm 2}$ of the two CUs for the system under consideration are given, for the block lengths of $\beta_1$ and $\beta_2$, as well as for the decoding error probabilities of $\epsilon_1$ and $\epsilon_2$, respectively, as:
\begin{subequations}
\begin{align}
R_\mathrm{1}=\ln{\left(1+\gamma_1\right)}-\sqrt{\frac{{{V_{1}}}}{{{\beta_{1}}}}} {Q^{ - 1}}\left({{\epsilon _{1}}} \right)\\
R_\mathrm{2}=\ln{\left(1+\gamma_2\right)}-\sqrt{\frac{{{V_{2}}}}{\beta_2}} {Q^{ - 1}}\left({\epsilon_2} \right),
\end{align}
\end{subequations}
where $\beta_1+\beta_2=N$. Let us assume that there exists a positive value $\hat{N}$ such that
$\hat{N}<N$. Based on this assumption, the corresponding point on the RBE-rate region should adhere 
to $\beta_1 + \beta_2 = \hat{N}$. Let us fix the block length of CU $1$ to $\beta_1$. Hence, the block length of CU 2 is given by $\hat{\beta}_2 = \hat{N}-\beta_1$. As a result, the achievable rate of CU 2 $\hat{R}_2$ for the block length of $\hat{\beta}_2$ is modified as follows
\begin{equation}
    \hat{R}_\mathrm{2}=\ln{\left(1+\gamma_2\right)}-\sqrt{\frac{{V_2}}{{\hat{\beta}_2}}}Q^{- 1}\left({{\epsilon_2}} \right).
\end{equation}
Note that the achievable rate of the SPC regime monotonically increases with block length \cite{SPC_1}. 
For the given $\hat{N} < N$, it follows that $\hat{\beta}_2=\hat{N}-\beta_1 < N-\beta_1$. Hence, $\hat{R}_\mathrm{2} < R_\mathrm{2}$, which reduces the achievable sum rate of the system. Therefore, for any achievable sum rate point on the Pareto boundary of the RBE-rate region, $\hat{N}$ must be equal to $N$. This holds for more than two CUs.
\end{proof}
Moreover, for any target rate $\mathcal{R} \geq 0$, the constraint (\ref{cons:rate}) can be modified as follows
\begin{equation}\label{mod:beta}
\begin{aligned}
    \beta_m \geq \left(\frac{\sqrt{{{V_{m}}}} {{Q^{ - 1}}\left({{\epsilon_{m}}} \right)}}{{\ln{\left(1+\gamma_m\right)}- {\eta_{m}}\mathcal{R}}}  \right)^2.
\end{aligned}
\end{equation}
Consequently, following Proposition \ref{Prop_2} and (\ref{mod:beta}), the modified block length optimization is given by
\begin{equation}\label{beta}
\begin{aligned}
&{\mathrm{Find:}} ~[\beta_1, \hdots, \beta_M] \\
&\mathrm{s. t.}\quad \text{(\ref{mod:beta}), (\ref{sum:N}), and (\ref{cons:IN})},
\end{aligned}
\end{equation}
where (\ref{beta}) is a mixed integer program for fixed $\mathbf{F}_\mathrm{RF}, \mathbf{F}_\mathrm{BB}$ and $\mathbf{U}$, which can be efficiently solved using the framework in \cite{MIP_1}. Finally, for a fixed BB TPC $\mathbf{F}_\mathrm{BB}$, RF TPC $\mathbf{F}_\mathrm{RF}$, and block length $\{\beta_m\}_{m=1}^M$, we update the achievable sum rate $\mathcal{R}$ via the bisection search method \cite{P_3}. The complete procedure of the proposed TLBS-based joint optimization of (\ref{OP:1}) is summarized in Algorithm \ref{alg:BBJ}.

\subsection{Computational complexity}
Since the inner layer of the proposed algorithm employs the BCD method for iteratively updating $\mathbf{F}_\mathrm{RF}, \mathbf{F}_\mathrm{BB}$ and $\mathbf{U}$ to minimize the RBE $\mathcal{E}$, in the $(\kappa + 1)$th iteration, we have
\begin{equation} 
\begin{aligned}
&\mathcal{E}(\mathbf{F}^{(\kappa + 1)}_\mathrm{RF}, \mathbf{F}^{(\kappa + 1)}_\mathrm{BB}, \mathbf{U}^{(\kappa + 1)})  \leq\mathcal{E}(\mathbf{F}^{(\kappa + 1)}_\mathrm{RF}, \mathbf{F}^{(\kappa + 1)}_\mathrm{BB}, \mathbf{U}^{(\kappa)}) \\ 
& \quad \leq \mathcal{E}(\mathbf{F}^{(\kappa + 1)}_\mathrm{RF}, \mathbf{F}^{(\kappa)}_\mathrm{BB}, \mathbf{U}^{(\kappa)}) \leq \mathcal{E}(\mathbf{F}^{(\kappa)}_\mathrm{RF}, \mathbf{F}^{(\kappa)}_\mathrm{BB}, \mathbf{U}^{(\kappa)}),
\end{aligned}
\end{equation}
where the RF TPC $\mathbf{F}^{(\kappa)}_\mathrm{RF}$ is optimized via the BMM and EPMO methods. Moreover, the BMM method returns progressively improved feasibility points with lower values of the objective function, while the EPMO method determines descent directions within the feasible region of the complex circle Riemannian manifold to achieve the same goal.
Furthermore, for the BB and RF TPCs designed, the outer layer subsequently optimizes the block length $\{\beta_m\}_{m=1}^M$ via solving (\ref{beta}) and updates the achievable sum rate via the bisection search method, until convergence is achieved.

We now evaluate the overall computational complexity of the proposed TLBS algorithm. In the inner layer, the complexity of computing the RF TPC $\mathbf{F}_\mathrm{RF}$ via the BMM and EPMO methods is $\mathcal{O}(\mathcal{I}_\mathrm{b}KN^2_{t}N^2_{RF})$ and $\mathcal{O}(\mathcal{I}_\mathrm{e}N^2_{t}N^2_{RF})$, respectively, where $\mathcal{I}_\mathrm{b}$ and $\mathcal{I}_\mathrm{e}$ are the number of iterations required to update $\mathbf{d}$ in the BMM and EPMO methods. Furthermore, the complexities involved in obtaining the BB TPC $\mathbf{F}_\mathrm{BB}$, the auxiliary matrix $\mathbf{U}$ and block length $\{\beta_m\}_{m=1}^M$ are given as $\mathcal{O}(N^{3.5}_{RF}M^{3.5})$, $\mathcal{O}(N_\mathrm{t}N_\mathrm{tar}N_\mathrm{RF})$ and $\mathcal{O}(\mathcal{I}_\mathrm{BL}M^2)$, respectively, where $\mathcal{I}_{BL}$ is the number of times problem (\ref{beta}) is solved out of the total number of $\mathcal{R}$ bisection iterations.
Therefore, the overall complexity of the TLBS algorithm harnessing the BMM and EPMO methods, namely TLBS-BMM and TLBS-EPMO, is given by \\ $\mathcal{O}\big[I_\mathrm{out}I_\mathrm{in}\left(\mathcal{I}_\mathrm{b}KN^2_{t}N^2_{RF}+N^{3.5}_{RF}M^{3.5}+N_\mathrm{t}N_\mathrm{tar}N_\mathrm{RF}\right)\big]+\mathcal{O}\left(\mathcal{I}_\mathrm{BL}M^2\right)$ and \\$\mathcal{O}\big[I_\mathrm{out}I_\mathrm{in}\left(\mathcal{I}_\mathrm{e}N^2_{t}N^2_{RF}+N^{3.5}_{RF}M^{3.5}+N_\mathrm{t}N_\mathrm{tar}N_\mathrm{RF}\right)\big]+\mathcal{O}\left(\mathcal{I}_\mathrm{BL}M^2\right)$, respectively, where $I_\mathrm{out}$ and $I_\mathrm{in}$ denote the number of iterations required in the outer and inner layers.
\begin{algorithm}[t]
\caption{Two layer bisection search (TLBS) algorithm for solving (\ref{OP:1})}
 \textbf{Input:}  $\mathbf{F}_\mathrm{r}$, $\{\eta_m\}_{m=1}^M$, $\{\epsilon_m\}_{m=1}^M$, $\mathcal{R}_\mathrm{L} = 0 $, $\mathcal{R}_\mathrm{U}$, $\mathcal{E}_\mathrm{max}$, $P_\mathrm{max}$, $N$, and thresholds $\tau_5>0$, $\tau_6>0$
 \label{alg:BBJ}
\begin{algorithmic}[1]
\State \textbf{initialize:} $\mathbf{F}_\mathrm{RF},\mathbf{F}_\mathrm{BB}, \mathbf{U}, \{\beta_m\}_{m=1}^M$, and $\mathcal{E}(\mathbf{U},\mathbf{F}_\mathrm{BB},\mathbf{F}_\mathrm{RF})$
         \Repeat
                \State $\mathcal{R} = (\mathcal{R}_\mathrm{L} + \mathcal{R}_\mathrm{U}) / 2$
                \State evaluate $\Gamma_{m},\forall m$ using (\ref{gamma})
                \Repeat 
                \State set $\kappa = 0$, $\mathcal{E}^{(\kappa)} = \infty$
                \State given $\mathbf{F}^{(\kappa)}_\mathrm{BB}$ and $\mathbf{U}^{(\kappa)}$, obtain $\mathbf{F}^{(\kappa + 1)}_\mathrm{RF}$ by solving (\ref{RF_1})
                \State given $\mathbf{F}^{(\kappa + 1)}_\mathrm{RF}$ and $\mathbf{U}^{(\kappa)}$, evaluate $\mathbf{F}^{(\kappa+1)}_\mathrm{BB}$ by solving (\ref{eqn:SOCP})                
                \State given $\mathbf{F}^{(\kappa + 1)}_\mathrm{RF}$ and $\mathbf{F}^{(\kappa+1)}_\mathrm{BB}$, calculate $\mathbf{U}^{(\kappa + 1 )}$ by solving (\ref{UTP})
                \State compute $\mathcal{E}^{(\kappa + 1)} = \mathcal{E}(\mathbf{U}^{(\kappa + 1)},\mathbf{F}^{(\kappa + 1)}_\mathrm{BB},\mathbf{F}^{(\kappa + 1)}_\mathrm{RF})$
                \State set $\kappa\leftarrow \kappa+1$
                \Until $|({\mathcal{E}^{(\kappa)}  - \mathcal{E}^{(\kappa - 1)} )}/\mathcal{E}^{(\kappa)} | \leq \tau_5 $
                \State \textbf{if} $\mathcal{E}^{(\kappa)} \leq \mathcal{E}_\mathrm{max}$ and $\|\mathbf{F}^{(\kappa)}_\mathrm{RF}\mathbf{F}^{(\kappa)}_\mathrm{BB}\|^2_F \leq P_\mathrm{max}$
                    \State \hspace{0.5cm} obtain $\{\beta_m\}_{m=1}^M$ using (\ref{beta})  
                    \State \hspace{0.5cm} \textbf{if} (\ref{beta}) is feasible
                    \State \hspace{1cm} set $\mathcal{R}_\mathrm{L} = \mathcal{R}$.
                    \State \hspace{0.5cm} \textbf{else} set $\mathcal{R}_\mathrm{U} = \mathcal{R}$.
                    \State \hspace{0.5cm} \textbf{end if}
                \State \textbf{else} set $\mathcal{R}_\mathrm{U} = \mathcal{R}$.
                \State \textbf{end if}
           \Until $ \mathcal{R}_\mathrm{U} - \mathcal{R}_\mathrm{L}  \leq \tau_6 $     
        \State \textbf{output:}  $\mathbf{F}_\mathrm{RF}$, $\mathbf{F}_\mathrm{BB}$, $\{\beta_m\}_{m=1}^M$ and RBE-rate tuple $\left(\mathcal{E}, \mathcal{R} \right)$
\end{algorithmic}
\end{algorithm}

\section{\uppercase{Simulation Results}}\label{simulation results}
In this section, our simulation results characterizing the Pareto boundaries of the RBE-rate region for various scenarios, together with the beam pattern, to demonstrate the performance of our proposed algorithms for an SPC-enabled mmWave ISAC system. The ISAC BS is assumed to have a ULA equipped with $N_\mathrm{t}$ transmit antennas and $N_\mathrm{RF}$ RF chains. Moreover, each CU and RT is assumed to be located within the range of $100$m from the ISAC BS having the path loss model $PL(d_m)$ for the mmWave channel, which is given by \cite{HBF_8}
\begin{equation}\label{eqn:path loss model}
\begin{aligned}
PL(d_m)\hspace{0.02in}[\rm dB] = \varepsilon + 10\varphi\log_{10}(d_m)+\varpi,
\end{aligned}
\end{equation}
where we have $\varpi \in {\cal CN}(0,\sigma_{\rm \varpi}^2)$ with $\sigma_{\rm \varpi}=5.8 \hspace{0.02 in}{\rm dB}$, $\varepsilon=61.4$ and $\varphi=2$ \cite{HBF_8}.
Additionally, we fix $N_\mathrm{clu} = 5$ and $N_\mathrm{ray} = 10$ with an angular spread of $10$ degrees to model the propagation environment. Furthermore, the AoDs $\phi_{i,j}, \forall i, j$ are generated from a Laplacian distribution and distributed uniformly within $\left[-90^\circ, 90^\circ\right]$. We consider two RTs and two CUs. Furthermore, the system operates at $28$ GHz with a bandwidth of $251.1886$ MHz. Thus, the noise variance $N_\mathrm{0}$ at each CU is set as $N_\mathrm{0}=-174+ 10 \log_{10}B=-90$ dBm. 
Unless otherwise stated, the key simulation parameters are those listed in Table \ref{Table2}. Moreover, all the simulation results are averaged over 100 channel realizations. 

\begin{table}[t]
\centering
\caption{Simulation Parameters and Corresponding Values} \label{tab:simulation parameters}
\begin{tabular}{l r}\label{Table2}\\
\hline
Parameter & value \\
\hline
Maximum allowable power budget $P_\mathrm{max}$ & ${30}$  dBm\\
Number of transmit antennas, $N_\mathrm{t}$ & 128\\
Number of RF chains, $N_\mathrm{RF}$ & \{4, 6\}\\
Maximum allowable block length, $N$ & \{128, 256\}\\
Decoding error probability of each CU , $\epsilon$ & $\{10^{-5}, 10^{-6}\}$\\
Number of uRLLC CUs & $2$\\
Number of RTs & $2$\\
Radar beam pattern error, $\mathcal{E}$ & $\{0.15, 0.45\}$ \\
Noise power, $N_{o}$ & $-90$ dBm \\
Target rate ratio of each CU, $\{\eta_1, \eta_2\}$ & $(0.5, 0.5)$ \\
\hline
\end{tabular}
\end{table}
\begin{figure*}[t]
\centering
\begin{subfigure}{.65\columnwidth}
\includegraphics[width=1.1\columnwidth]{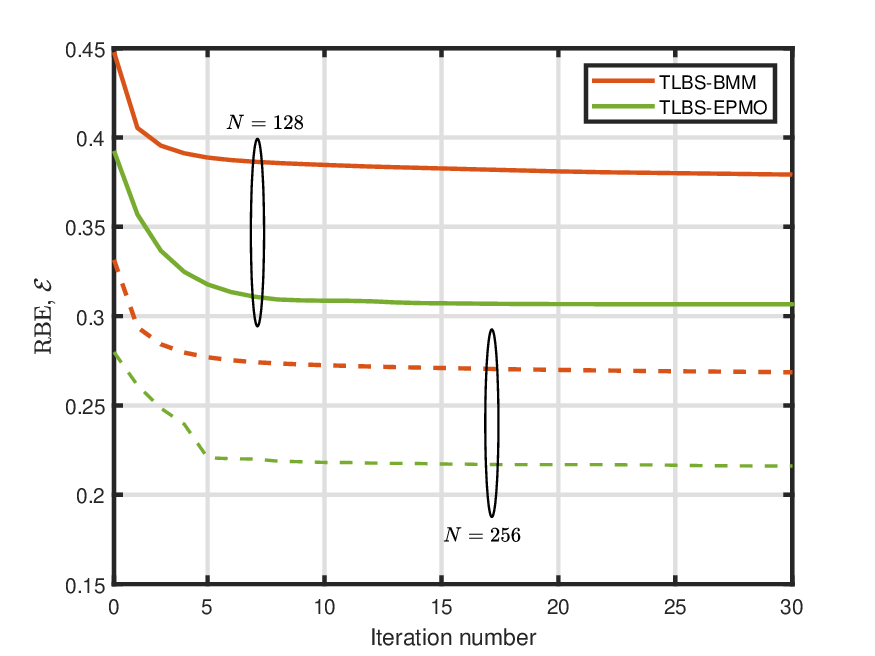}%
\caption{}
\label{fig:R1}
\end{subfigure}
\begin{subfigure}{.65\columnwidth}
\includegraphics[width=1.1\columnwidth]{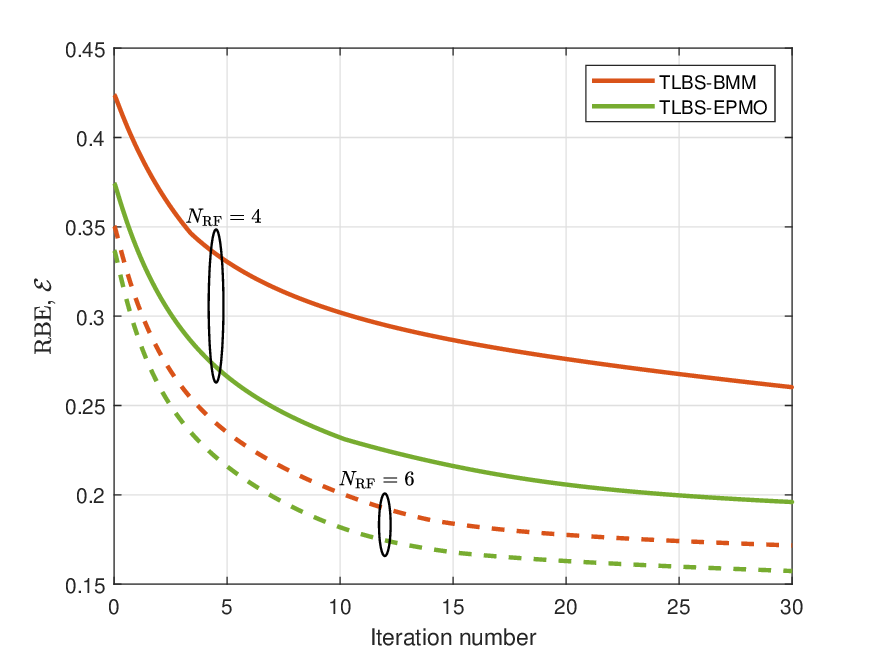}%
\caption{}
\label{fig:R2}
\end{subfigure}%
\begin{subfigure}{.65\columnwidth}
\includegraphics[width=1.1\columnwidth]{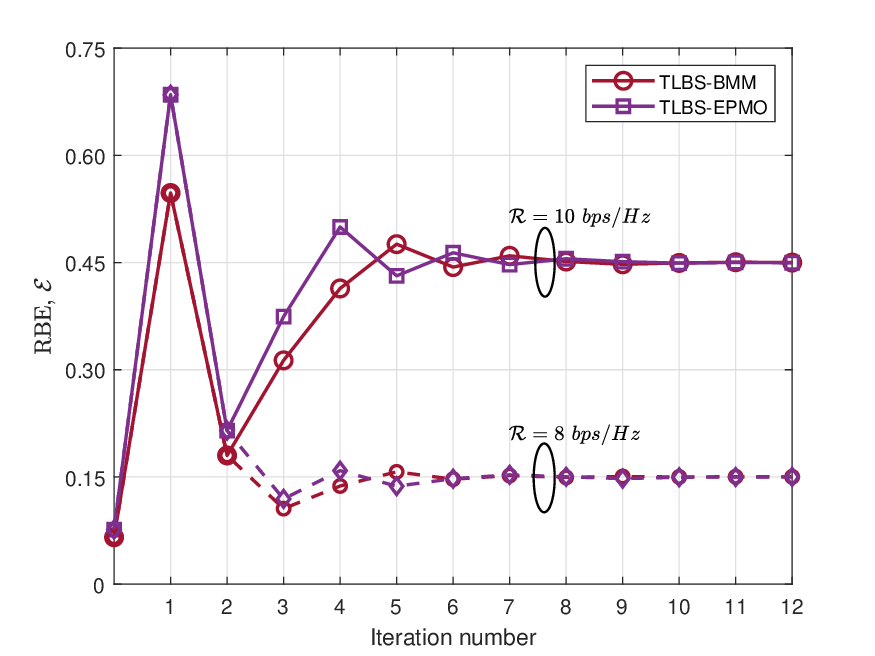}%
\caption{}
\label{fig:R3}
\end{subfigure}%
\caption{Convergence of RBE in the inner layer for fixed $\mathcal{R}=10$ bps/Hz (a) with $N$; 
(b) with $N_\mathrm{RF}$; (c) convergence of RBE in the outer layer with $N=128$ and $N_\mathrm{RF}=4$.}
\vspace{-5mm}
\end{figure*}

\subsection{Convergence behavior of the proposed algorithms}
In this subsection, we characterize the convergence behavior of the proposed TLBS-based joint optimization Algorithm \ref{alg:BBJ} to solve (\ref{eqn:HBF_ISAC_6}). Fig. \ref{fig:R1} and Fig. \ref{fig:R2} present the convergence of RBE $\mathcal{E}$ in the inner layer of the proposed algorithms with respect to the block length $N$, the number of RFCs, $N_\mathrm{RF}$, respectively, for a fixed sum rate of $\mathcal{R}=10$ bps/Hz. Furthermore, we compare the convergence performance of the proposed TLBS-BMM and TLBS-EPMO. One can observe from both figures that the RBE of the TLBS-BMM and TLBS-EPMO methods decreases monotonically, verifying the convergence of Algorithm \ref{alg:BBJ} in the inner layer.
Moreover, the RBE of the TLBS-EPMO method is much lower than that of the TLBS-BMM method for a fixed sum rate $\mathcal{R}$. This is due to the fact that the TLBS-BMM involves the approximation of the majorizer functions, which increases the RBE. However, the TLBS-EPMO method does not require such approximations, which reduces the gap between the optimal radar beamformer and the HBF designed. Furthermore, the RBE corresponding to $N=256$ is lower than that of $N=128$ since, upon increasing the block length, the SPC rate approaches the Shannon capacity, and hence for the given target rate, more transmit power is available towards the RTs. In a similar fashion, the RBE of $N_\mathrm{RF}=6$ is much less than that of $N_\mathrm{RF}=4$, which is due to the fact that increasing $N_\mathrm{RF}$ improves the approximation of HBF for the ideal radar beamformer.

Fig. \ref{fig:R3} shows the convergence of the proposed algorithm in the outer layer comprising a binary search approach, for $\mathcal{R}=\{8, 10\}$ bps/Hz. Observe that the proposed algorithm converges within $10$ iterations for both BMM and EPMO techniques, which evidences the convergence of the proposed TLBS algorithm. Moreover, the RBE is higher for $\mathcal{R}=10$ bps/Hz than $\mathcal{R}=8$ bps/Hz, which is due to the fact that a large sum rate requirement for the CUs reduces the power radiated towards the RTs.

Furthermore, to demonstrate the efficiency of the proposed algorithms and to glean interesting design insights, we compare the proposed method to the following schemes.
\begin{itemize}
    \item \textit{Scheme 1} (Optimal IBL-FDB): For this scheme, IBL is employed at the ISAC BS, which follows the Shanon capacity (SC) (\ref{eqn:shanon}). Furthermore, FDB is used for designing the beamformer.
    \item \textit{Scheme 2} (TLBS-FDB): This scheme corresponds to SPC transmission along with the FDB scheme to design the beamformer, where the TLBS algorithm is employed for optimizing the FDB and block length.
    \item \textit{Scheme 3} (TLBS-OMP): In this scheme, we employ the orthogonal matching pursuit (OMP) \cite{mm_1} in the inner layer of the TLBS algorithm to optimize the RF and BB TPCs.
\end{itemize}
We compare the performance by evaluating the RBE-rate region and the sum rate versus several important parameters, which are discussed in the subsequent subsections.

\begin{figure*}[t]
\centering
\begin{subfigure}{.65\columnwidth}
\includegraphics[width=1.1\columnwidth]{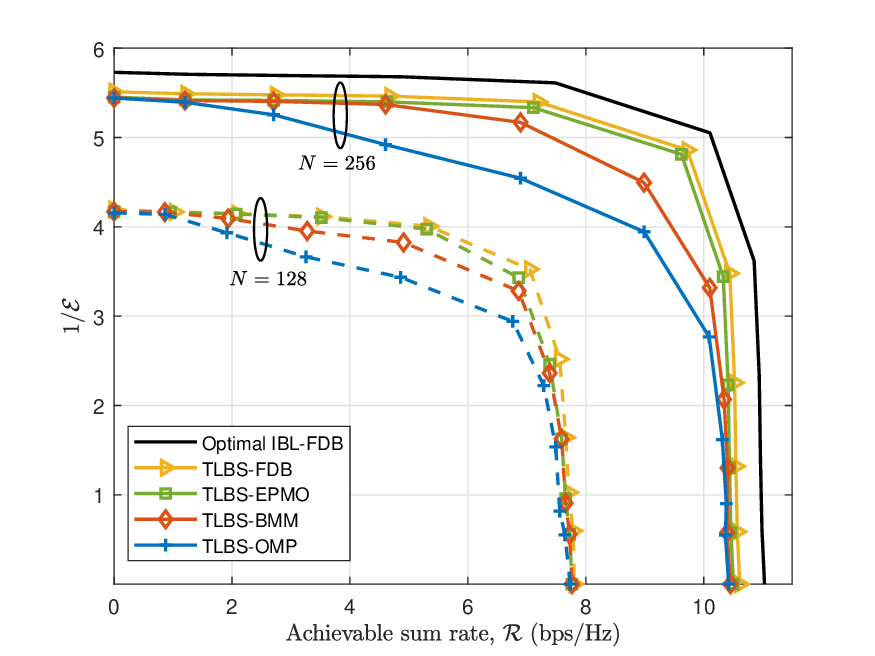}%
\caption{}
\label{fig:R4}
\end{subfigure}
\begin{subfigure}{.65\columnwidth}
\includegraphics[width=1.1\columnwidth]{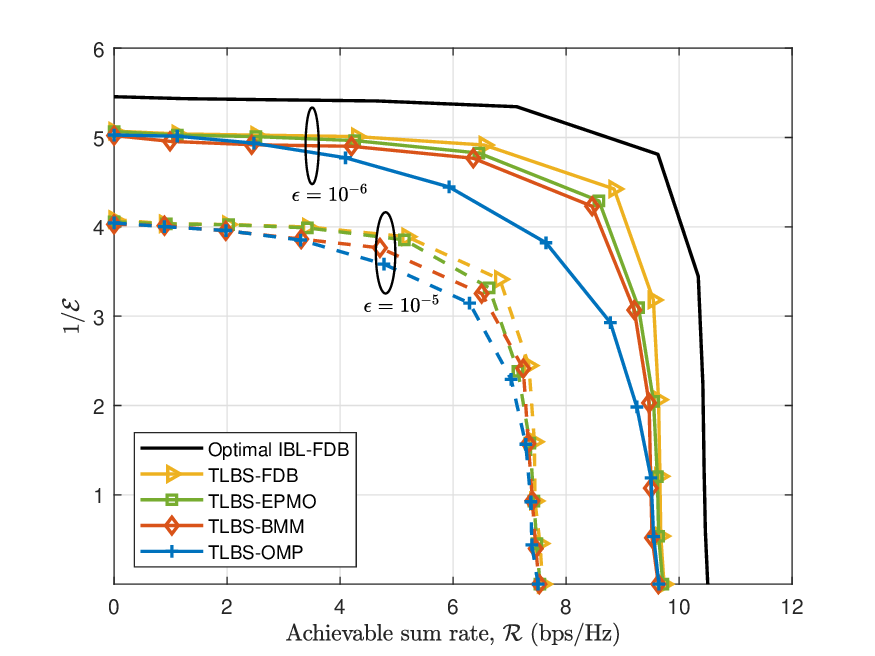}%
\caption{}
\label{fig:R6}
\end{subfigure}%
\begin{subfigure}{.65\columnwidth}
\includegraphics[width=1.1\columnwidth]{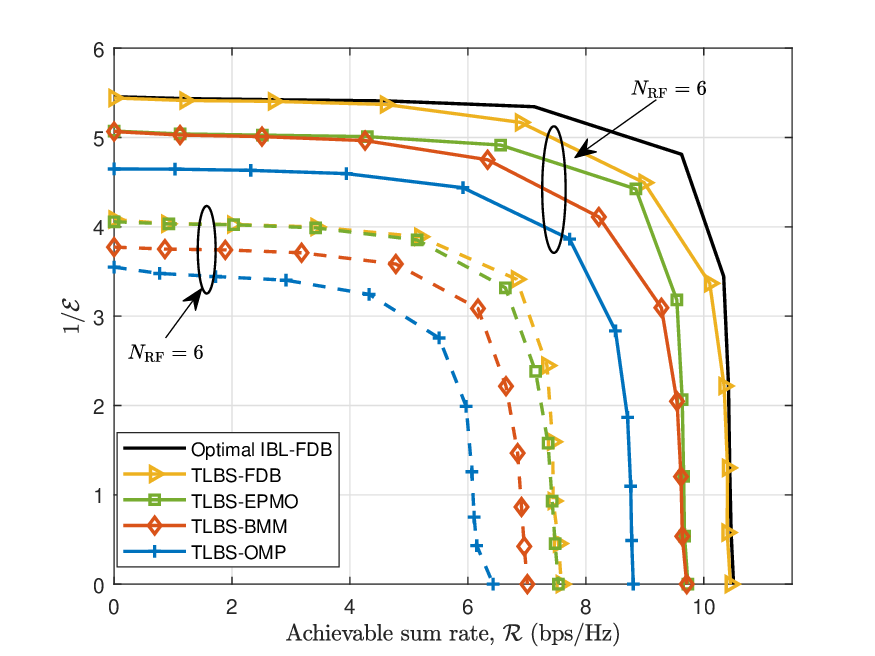}%
\caption{}
\label{fig:R5}
\end{subfigure}%
\caption{Pareto boundary of RBE-rate region for different (a) block length $N$; (b) decoding error probabilities $\epsilon$; (c) number of RFCs $N_\mathrm{RF}$.}
\vspace{-5mm}
\end{figure*}

\subsection{Pareto boundary of the RBE-rate region}
In this subsection, we investigate the behavior of the Pareto boundary of the RBE-rate region in SPC-enabled mmWave MIMO ISAC systems by varying some important parameters. 
\subsubsection{Pareto boundary of RBE-rate region for different block lengths $N$}
In Fig. \ref{fig:R4}, we plot the Pareto boundary of the RBE-rate region for block lengths of $N =128$ and $256$ at a fixed decoding error probability\footnote{Note that $\epsilon$ represents the decoding error probability due to SPC, whereas $\mathcal{E}$ is the RBE.} of $\epsilon=10^{-5}$ when the number of RFCs is $N_\mathrm{RF} = 4$. As seen from the figure, the Pareto boundary of the RBE-rate region increases with $N$, since a larger $N$ results in a higher sum rate, which reveals the impact of the block length on the system due to the SPC transmission.  
Moreover, the IBL-FDB scheme serves as the global upper bound for the RBE-rate region due to the resultant gain of the IBL transmission coupled with the FDB scheme. Meanwhile, TLBS-FDB acts as the local upper bound for the proposed schemes in the SPC regime for both $N=128$ and $N=256$ due to the FDB scheme.
Furthermore, the TLBS-EPMO scheme yields improved performance over the TLBS-BMM scheme and it is close to the locally optimal curve of the TLBS-FDB for both $N=128$ and $256$, which shows the efficacy of the RCG approach in the context of the EPMO technique. Moreover, both the proposed TLBS-EPMO and TLBS-BMM schemes are clearly superior to the TLBS-OMP method, which shows the effectiveness of the EPMO and BMM methods in optimizing the RF TPC. 
\begin{figure*}[t]
\centering
\begin{subfigure}{.65\columnwidth}
\includegraphics[width=1.1\columnwidth]{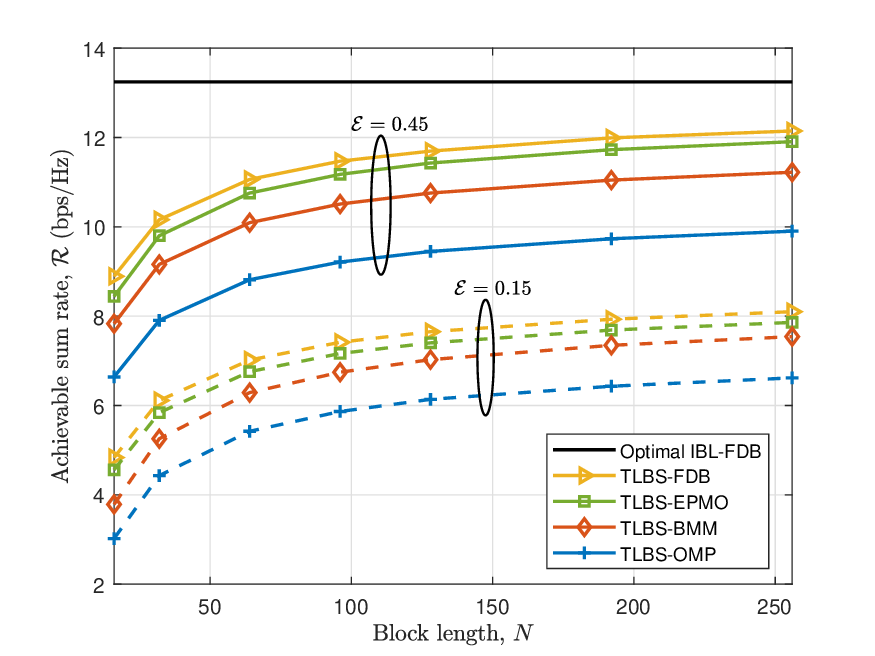}%
\caption{}
\label{fig:R7}
\end{subfigure}
\begin{subfigure}{.65\columnwidth}
\includegraphics[width=1.1\columnwidth]{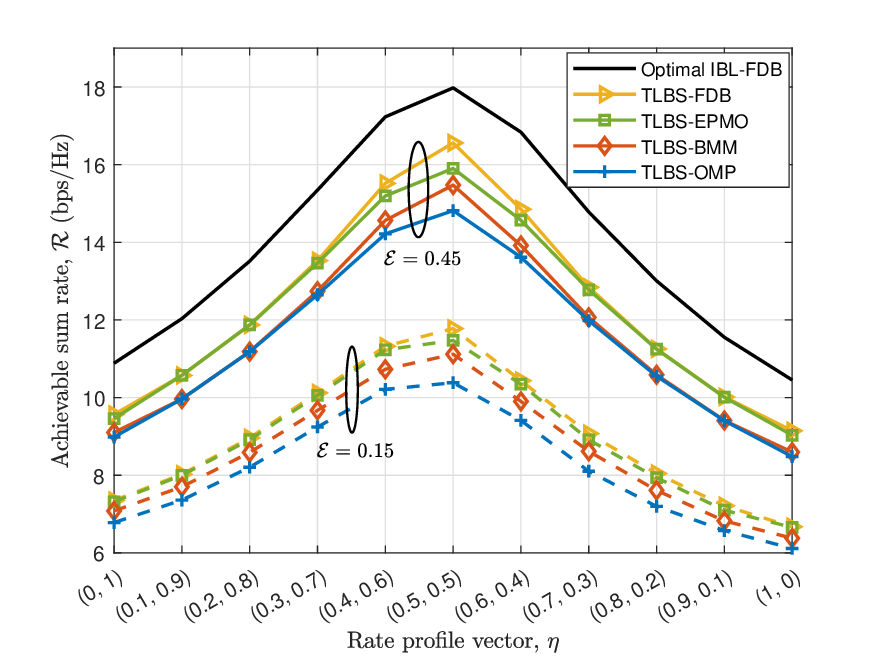}%
\caption{}
\label{fig:R9}
\end{subfigure}%
\begin{subfigure}{.65\columnwidth}
\includegraphics[width=1.1\columnwidth]{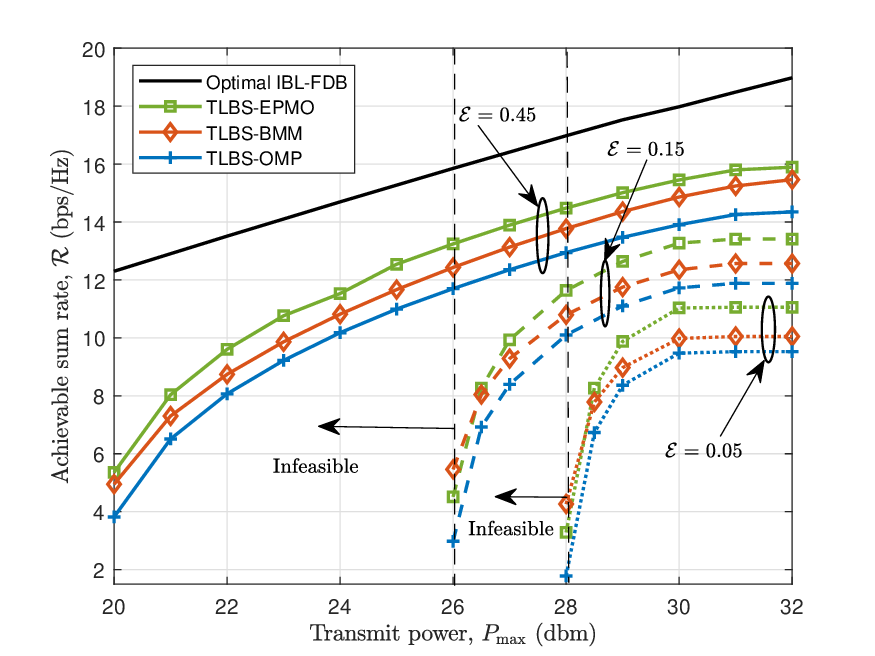}%
\caption{}
\label{fig:R11}
\end{subfigure}%
\caption{Achievable sum rate versus (a) block length $N$; (b) rate profile $\boldsymbol{\eta}$; (c)  transmit power $P_\mathrm{max}$. }
\vspace{-5mm}
\end{figure*}

\subsubsection{Pareto boundary of the RBE-rate region for different decoding error probabilities $\epsilon$}
Fig. \ref{fig:R6} investigates the impact of decoding error probability on the Pareto boundary of the RBE-rate region. As seen from the figure, the gap in the RBE-rate region increases as the decoding error probability decreases from $\epsilon=10^{-5}$ to $10^{-6}$ dB since a reduction in the decoding error probability of the SPC regime results in an increase in the achievable rate. Thus, upon decreasing the decoding error probability, the power available for the RTs increases for a given sum rate, leading to an RBE reduction. Furthermore, the Pareto boundary of the proposed schemes is superior to that of the TLBS-OMP method for both $\epsilon=10^{-5}$ and $10^{-6}$ dB, which shows the efficacy of the MM and RCG steps employed in the BMM and EPMO algorithms, respectively.
\subsubsection{Pareto boundary of the RBE-rate region for different RFCs $N_\mathrm{RF}$}
Fig. \ref{fig:R5} reveals the Pareto boundary of the RBE-rate region for $N_\mathrm{RF} = \{4, 6\}$  along with $N = 128$ and $\epsilon=10^{-5}$. It can be observed from the figure that the Pareto boundary of the RBE-rate region expands upon increasing the values of $N_\mathrm{RF}$. This can be explained by the fact that the error between the ideal radar beamformer and the HBF designed decreases upon increasing $N_\mathrm{RF}$, which therefore results in a reduced RBE. Consequently, more power is available for the CUs for a given RBE, which in turn leads to an increase in the sum rate.
Furthermore, the Pareto boundary of the proposed TLBS-BMM and TLBS-EPMO methods approaches that of the locally optimal TLBS-FDB for both $N_\mathrm{RF}=4$ and $6$. This shows that our proposed methods in the SPC regime achieve optimal performance with fewer RFCs. Hence, the proposed schemes save power and cost by employing the HBF scheme, while still achieving a performance that is close to that of the optimal FDB scheme.
Moreover, the Pareto boundary of the proposed TLBS-EPMO scheme is very close to that of the globally optimal IBL-FDB technique. Therefore, one can approach the Shannon capacity of the SPC-enabled mmWave MIMO ISAC system at a fixed block length and decoding error probability by increasing the number of RFCs $N_\mathrm{RF}$ in the TLBS-EPMO approach.   
\subsection{Achievable sum rate of the SPC-enabled mmWave MIMO ISAC systems}
\subsubsection{Achievable sum rate versus block length $N$}

In Fig. \ref{fig:R7}, we plot the achievable sum rate versus the block length for different RBEs $\mathcal{E} = \{0.15, ~0.45\}$. It can be seen from the figure that the IBL-FDB scheme is independent of the block length and acts as the global optimum.
Moreover, the achievable sum rate increases upon increasing the block length $N$ due to the influence of the block length on the rate expression. Furthermore, the sum rates of the proposed TLBS-EPMO and TLBS-BMM schemes approach that of the locally optimal TLBS-FDB, are seen to be improved over the TLBS-OMP scheme for increasing $N$, which shows the efficacy of the proposed designs.
In addition, one can observe from the figure that reducing the RBE from $\mathcal{E}=0.45$ to $0.15$ results in a decrease in the achievable sum rate. This is due to the fact that reducing the RBE results in an increased focus on the RTs, leading to a reduced sum rate, as expected.
Moreover, the TLBS-EPMO scheme has a performance edge over its TLBS-BMM counterpart for both $\mathcal{E}=0.15$ and $\mathcal{E} = 0.45$, which is due to the RCG step involved in the TLBS-EPMO approach conceived for the design of the RF beamformer.  
\subsubsection{Achievable sum rate versus the rate profile vector $\mathbf{\eta}$}
Fig. \ref{fig:R9} investigates the impact of the rate profile vector $\boldsymbol{\eta}$\footnote{In case of two CUs, rate profile vector $\boldsymbol{\eta}$ is given by $\boldsymbol{\eta} = [\eta_1, \eta_2]$, where $\eta_1$ and $\eta_2$ represents the target rate ratio of CU $1$ and $2$, respectively, with $\eta_1+\eta_2=1$.} on the achievable sum rate of the system for the RBEs of $\mathcal{E} = \{0.15, ~0.45\}$. As discussed, the elements of $\boldsymbol{\eta}$ denote the target ratio of the $m$th CU rate and to the sum rate of the system and satisfy the constraint $\sum_{m=1}^M \eta_m = 1$ as associated with $\eta_m \in (0,~1)$. Therefore, in pair of the uRLLC CUs, we set $\boldsymbol{\eta}$ as $\boldsymbol{\eta} = [\eta, ~1-\eta_]$ and vary $\eta$ from $0$ to $1$ with increments of $0.1$. As seen from the figure, the achievable sum rate of the system is quasi-concave in nature with respect to the rate profile vector. Therefore, an optimal value of the rate profile vector exists at which the achievable rate is maximum. 
\subsubsection{Achievable sum rate versus transmit power $P_\mathrm{max}$}

We plot the achievable sum rate versus the transmit power in Fig. \ref{fig:R11} for the fixed RBEs of $\mathcal{E} = \{0.05, 0.15, 0.45\}$. For a fixed RBE, the transmit power is a feasibility parameter for the TLBS algorithm. Therefore, as seen from the figure, for the RBE values of $\mathcal{E}=0.05$ and $\mathcal{E}=0.15$, the TLBS algorithm is infeasible for $P_\mathrm{t}<28$ dBm and $P_\mathrm{t}< 26$ dBm, respectively. However, a large RBE of $\mathcal{E}=0.45$ is always feasible for the TLBS algorithm. This trend is due to the fact that a small RBE leads to focusing a large fraction of the available power for the RTs. Hence, the power radiated towards the CUs is low, which renders the problem infeasible due to the inability of achieving their QoS requirement. Moreover, the achievable sum rate of the system increases with the transmit power, and the proposed schemes yield an improved sum rate in comparison to the TLBS-OMP method.
\subsection{Beampattern of the SPC-enabled mmWave MIMO ISAC system}
For this scenario, we consider the RTs and uRLLC CUs to be located at $[-60^\circ, -20^\circ]$ and $[30^\circ, 60^\circ]$, respectively. Therefore, the desired beampattern is given by
\begin{equation}\label{eqn:desired_beam}
G_{\mathrm{d}}(\theta_{l})= {\begin{cases}1, ~\theta_{l}\in (\overline{\theta }_i\pm\sigma_{\theta}),~i=1, 2,\\ 0,~\text{otherwise},
\end{cases}}
\end{equation}
where $\overline{\theta}$ is the direction of the target and $\sigma_{\theta}$ denotes a constant angular spread of $\sigma_\theta$, which is assumed to be $\frac{1}{\sqrt{2}}$. Fig. \ref{fig:R13} shows the ideal beam pattern of the radar- and communication-only beamformer. As seen from the figure, the main lobes of the beam pattern are directed towards the location of the RTs and the communicating CUs. Furthermore, Fig. \ref{fig:R12} plots the beam pattern of the proposed HBF schemes and compares it to the baseline for block lengths of $N = 128$ and $256$ for a fixed RBE of $\mathcal{E} = 0.15$ and sum rate of $\mathcal{R}=10$ bps/Hz. As seen from the figure, the main lobes of the HBF beampattern are directed toward the RTs and the CUs. Moreover, the beamforming gain of the system toward the RTs is higher for $N=256$ than $N=128$. This is due to the fact that a large $N$ increases the sum rate. As a result of this, higher power is available for the target at a given sum rate and RBE. Additionally, the gain of the proposed schemes is higher than that of the TLBS-OMP method, which once again vindicates the efficacy of the EPMO and BMM algorithms conceived.

\begin{figure}[t]
\setkeys{Gin}{width=\linewidth}
    \begin{subfigure}[t]{0.24\textwidth}
    \includegraphics[width=1.1\textwidth]{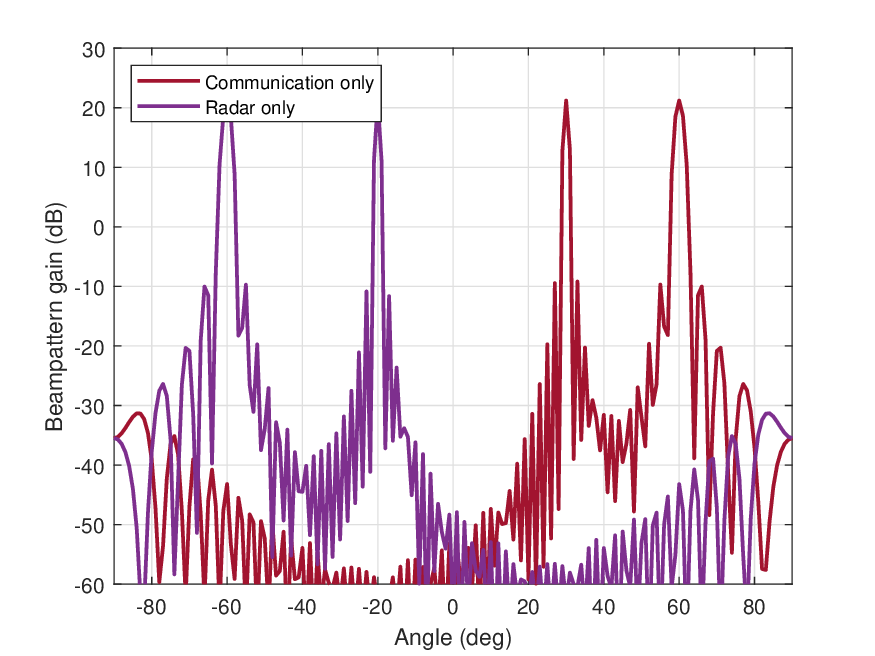}
    \caption{} \label{fig:R13}
\end{subfigure}
\begin{subfigure}[t]{0.24\textwidth}
    \includegraphics[width=1.1\textwidth]{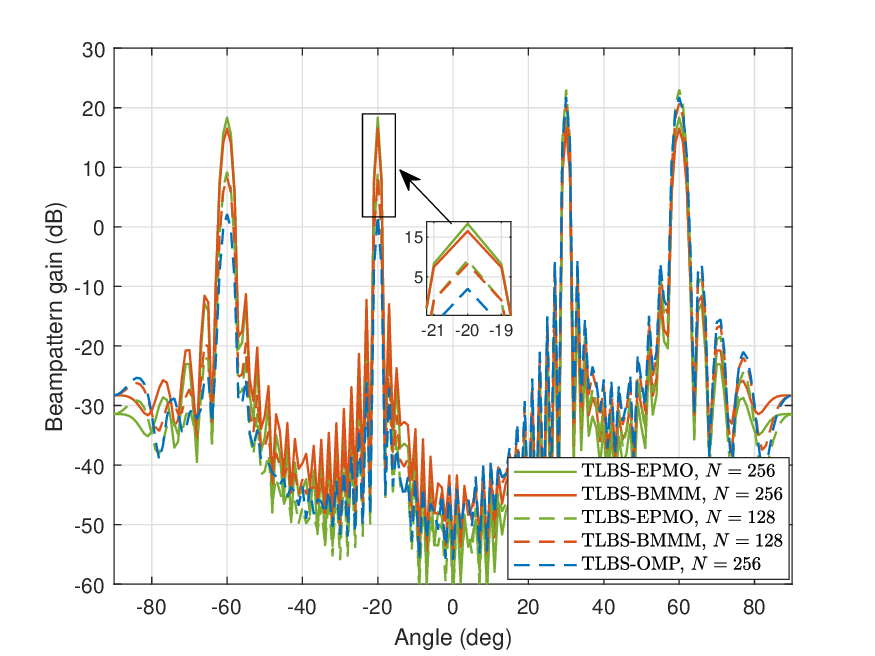}
    \caption{} \label{fig:R12}
\end{subfigure}
\caption{Beam pattern for fixed $\mathcal{R}=10$ bps/Hz (a) with radar and communication only; 
(b) with HBF}
\vspace{-0.7cm}
\end{figure}

\section{\uppercase{Conclusion}}\label{conclusion}
Pareto-optimal joint HBF and block length designs were conceived by considering the SPC transmission in the mmWave ISAC systems to meet the uRLLC requirements of the CUs, while also accomplishing sensing of the RTs. To this end, a Pareto-optimal framework was developed for characterizing the RBE-rate region of the model considered via the optimization of the RF and BB TPCs, as well as the block lengths. A novel TLBS algorithm was proposed for HBF design that comprises two layers. The inner layer computed the RF and BB TPCs minimizing the RBE of the RTs for a fixed sum rate of the CUs. Subsequently, the outer layer achieved block length optimization and evaluated the sum rate achievable for the given RBE. As a further advance, a pair of algorithms were proposed to design the RF TPC for the given system, namely, the BMM and EPMO schemes, which are based on the MM and RCG principles, respectively. Finally, simulation results were presented for characterizing the complete RBE-rate region and the sum rate of the system achievable for various parameter settings. The results evidence the fact that, through careful design, the mmWave ISAC system relying on SPC achieves the performance of an ideal IBL-aided mmWave ISAC system, despite using substantially fewer RFCs and a finite block length. Thus, the proposed design is cost- and power-efficient, while supporting uRLLC services in 6G ISAC mmWave systems.

\begin{appendices}
\section{\uppercase{Proof of Proposition}} \label{Appendix:A}
Taking the constraint (\ref{cons:rate}) into account with equality, we have
\begin{equation}\label{Prop:3}
\ln{\left(1+\gamma_m\right)}-\sqrt{\frac{{V_{m}}}{\beta_m}}Q^{- 1}\left({\epsilon_m} \right) = \eta_{m}\mathcal{R}. 
\end{equation}
By employing ${\delta_{m}}$ and $\tau_m =  \frac{Q^{ - 1}\left({{\epsilon _{m}}} \right)}{\sqrt{\beta_{m}}}$, (\ref{Prop:3}) can be rewritten as 
\begin{equation}\label{Prop:4} 
\ln{\big[\delta_m\left(1+\gamma_m\right)\big]} = \sqrt{{V_{m}}}\tau_m.
\end{equation}
Defining $\varrho_{m}= \ln{\big[\delta_m\left(1+\gamma_m\right)\big]}$ leads to
\begin{equation}\label{Prop:5}
\delta^2_me^{-2\varrho_m} - \frac {1}{\left ({1+\gamma_m }\right )^{2}} = 0.
\end{equation}
From (\ref{Prop:5}), we obtain $V_{m} = 1 - \delta_m^2e^{-2\varrho_m}$. Substituting this into (\ref{Prop:4}) and rearranging yields
\begin{equation}\label{Prop:6}
\varrho_m = \sqrt{{1 - \delta_m^2e^{-2\varrho_m}}}\tau_m.\
\end{equation}
Considering $\kappa_m = 2\varrho_m$, and applying basic mathematical operations, (\ref{Prop:6}) may be transformed as follows:
\begin{equation}\label{Prop:7}
e^{\kappa_{m}}\left ({\kappa_{m}-2\tau_m}\right )\left ({\kappa_{m}+2\tau_m }\right )=-4\delta ^{2}_m\tau_m^{2}.
\end{equation}
Note that (\ref{Prop:7}) is a well-known transcendent equation \cite{SPC_2}, whose solution can be obtained by the generalized Lambert ${\mathcal W}$  function. Consequently, the minimum value of $\gamma_{m}$ can be achieved by $\Gamma_{m}=e^{\eta_{m}R + \frac{\kappa_{m}}{2}} - 1$.
\section{PROOF OF THEOREM 1}\label{Appendix:B}
When $g(\mathbf {x}) = \mathbf {x}^{H}\mathbf {Q}\mathbf {x}$, the following inequality is proved in \cite{sun2016majorization}
\begin{equation} 
\begin{aligned}
\mathbf{x}^{H}\mathbf {Q}\mathbf {x} \leq &2{\rm Re}\left({\mathbf x}^{H}\left({\mathbf Q}-{\mathbf R}\right)\overline{\mathbf {x}}\right)\\ &{\mathbf x}^{H}{\mathbf R}{\mathbf x}+\overline{\mathbf {x}}^{H}\left({\mathbf R}-{\mathbf Q}\right)\overline{\mathbf {x}}. \label{eq:MM_1}
\end{aligned}
\end{equation}
Here, $\mathbf {x}$ and $\overline{\mathbf {x}}$ are vectors in the domain of $g$,  and ${\mathbf R} \succcurlyeq {\mathbf Q}$, where ${\mathbf Q}$ is a Hermitian matrix, with equality is achieved when $\mathbf {x} = \overline{\mathbf {x}}$.
The right-hand side of equation (\ref{eq:MM_1}) represents the majorant function of the quadratic form $g(\mathbf{x})$. According to \cite{he2022qcqp}, a majorizer of the quadratic $g(\mathbf{x})$
is constructed as:
\begin{equation}
\begin{aligned}
\mathbf {x}^{H}\mathbf {Q}\mathbf {x}\leq & 2\mathrm{{Re}}\left(\mathbf {x}^{H}\left(\mathbf {Q}-t\mathbf {I}\right)\overline{\mathbf {x}}\right) \\ & +t\mathbf {x}^{H}\mathbf {I}\mathbf {x}+ \overline{\mathbf {x}}^{H}(t\mathbf {I}-\mathbf {Q})\overline{\mathbf {x}} \label{eq:MM_2},
\end{aligned}
\end{equation}
where $t =  $ $\mathrm{tr}\left(\mathbf {Q}\right) $ or $\mathbf{\lambda}_{max}\left(\mathbf {Q}\right) $ and the choice depends on finding a balance between the computational complexity and convergence speed. Under the UM constraints $\left\vert\mathbf{x}(l)\right\vert = \left\vert\overline{\mathbf {x}}(l)\right\vert = 1$, $\mathbf {x}^{H}\mathbf {I}\mathbf {x}$ and $\mathbf \overline{\mathbf {x}}^{H}\mathbf {I}\overline{\mathbf {x}}$ equals its dimension, say $D$. Then (\ref{eq:MM_2}) reduces to:
\begin{equation}
\mathbf {x}^{H}\mathbf {Q}\mathbf {x}\leq  2\mathrm{{Re}}\left(\mathbf {x}^{H}\left(\mathbf {Q}-t \mathbf {I}\right)\overline{\mathbf {x}}\right)  -\overline{\mathbf {x}}^{H}\mathbf {Q}\overline{\mathbf {x}} + 2tD. \label{eq:MM_3}
\end{equation}
Similarly the majorizer function of the form $\mathbf{x}^{H}\mathbf{Q}\mathbf{x} - 2{{\mathrm{Re}}\left(\mathbf{P}^{H}\mathbf{x}\right)} + C$ under the UM constraints is given by
\begin{equation}\label{eq:MM_4} 
\!2\mathrm{{Re}}\left(\mathbf {x}^{H}[(\mathbf {Q} \!-\! t\mathbf {I}) \overline{\mathbf {x}} \!-\! \mathbf {P} ]\right) \!-\overline{\mathbf {x}}^{H} \mathbf {Q}\overline{\mathbf {x}}+\!2tD\!+\! C. \end{equation}

\end{appendices}
\bibliographystyle{IEEEtran}
\bibliography{biblio.bib}
\end{document}